\documentclass[11pt]{article}
\usepackage{subfiles}
\usepackage[english]{babel}
\usepackage[a4paper,top=3 cm,bottom=3 cm,inner=2.5 cm,outer=2.5 cm]{geometry} 
\usepackage{graphicx}
\usepackage{bm}
\usepackage{mathtools}
\usepackage{amsmath}
\usepackage{amsfonts}
\usepackage{subcaption}
\usepackage{mathbbol}
\usepackage{enumitem}
\usepackage{cancel}

\setlist[itemize]{label=\textbullet}
\usepackage{float} 
\usepackage[table]{xcolor} 
\usepackage{booktabs} 
\usepackage[utf8]{inputenc} 
\usepackage{hyperref}
\usepackage{amsmath, amssymb, amsthm} 
\usepackage{eurosym}
\usepackage{tikz}
\theoremstyle{plain}
\newtheorem{theorem}{Theorem}[section] 
\newtheorem{lemma}[theorem]{Lemma}
\newtheorem{proposition}[theorem]{Proposition}
\newtheorem{definition}[theorem]{Definition}
\newtheorem{corollary}[theorem]{Corollary}

\theoremstyle{definition}
\usepackage{wasysym}
\hypersetup{pdfborder={0 0 0}} 
\usepackage{url} 
\usepackage{titlesec}

\titleformat{\paragraph}
{\normalfont\normalsize\bfseries}{\theparagraph}{1em}{}
\titlespacing*{\paragraph}
{0pt}{3.25ex plus 1ex minus .2ex}{1.5ex plus .2ex}

\titleformat{\subparagraph}
{\normalfont\normalsize\bfseries}{\theparagraph}{1em}{}
\titlespacing*{\paragraph}
{0pt}{3.25ex plus 1ex minus .2ex}{1.5ex plus .2ex}

\usepackage{mathrsfs}

\usepackage{braket}
\usepackage{physics}

\newcommand{\supp}[1]{\text{supp}(#1)}
\newcommand{\dvol}[1]{\mathrm{d}#1}

\newcommand{\WF}[1]{{\text{WF}{#1}}}

\usepackage{csquotes}
\usepackage[
backend=bibtex,
style=numeric, 
sorting = nyt,
isbn = false,
url = false,
maxbibnames = 4,
giveninits = true,
date=year,
]{biblatex}
\renewbibmacro{in:}{}
\AtEveryBibitem{%
  \clearfield{number}%
  \clearfield{issue}%
  \clearfield{month}}

\begin{document}
\par 
\bigskip 
\LARGE 
\noindent 
\textbf{Regularized interacting scalar quantum field theories} 
\bigskip \bigskip
\par 
\rm 
\normalsize 
 
\large
\noindent 
{\bf Nicola Pinamonti$^{1,2,a}$} \\
\par
\small

\noindent$^1$ Dipartimento di Matematica, Universit\`a di Genova - Via Dodecaneso, 35, I-16146 Genova, Italy. \smallskip

\noindent$^2$ Istituto Nazionale di Fisica Nucleare - Sezione di Genova, Via Dodecaneso, 33 I-16146 Genova, Italy. \smallskip

\smallskip
\smallskip

\noindent E-mail: 
$^a$nicola.pinamonti@unige.it,

\normalsize
${}$ \\ \\
 {\bf Abstract} \ \ 
In this paper we consider self-interacting scalar  quantum  field theories over a $d$-dimensional Minkowski spacetime with various interaction Lagrangians which are suitable functions of the field.
The interacting field observables are represented as power series over the free theory by means of 
perturbation theory.
The object which is employed to obtain this power series is
the time-ordered exponential of the interaction Lagrangian which is the $S$-matrix of the theory and thus itself a power series in the coupling constant of the theory.
We analyze a regularization procedure which makes the $S$-matrix
convergent to well-defined unitary operators.
This regularization depends on two parameters. One describes how much the high-frequency contributions in the propagators are tamed and a second one describes how much the large-field contributions are suppressed in the interaction Lagrangian. 
{We finally discuss the limits in which the parameters are removed in some lower-dimensional theories and for specific interaction Lagrangians.}
In particular, we show, {by an abstract argument}, that in three spacetime dimensions for a $\phi^4_3$ theory one obtains sequences of unitary operators which are weakly-$*$ convergent 
to {bounded operators} 
in the limit of vanishing parameters. 
{Although we lose control over the unitarity of the operator obtained in the limit and although the uniqueness of the limit is not under control, 
we notice that some basic causal properties are satisfied by the corresponding interacting fields.
We furthermore observe that the very same limits, applied at fixed order, give back the ordinary perturbation theory results.} 
Finally, we discuss a possible extension of these results to the case of a $\phi^4_4$ theory.

\color{black}

\bigskip
${}$

\tableofcontents

\section{Introduction}

{\bf Motivations and outlook}
\\
The concrete complete construction of (self)interacting quantum field theories is currently available only in a few cases or for low-dimensional spacetimes, see \cite{GlimmJaffe, Rivasseau}. 
In the generic case, interacting fields can nevertheless be treated as formal power series in the coupling parameters.
Predictions obtained from truncated perturbative series are in extremely good agreement with observations. One of the most celebrated successes of quantum field theories treated with perturbative methods is the prediction of the fine-structure constant of quantum electrodynamics which agrees with experiments up to a very high level of precision. 
Nevertheless, a full understanding of the theoretical and mathematical aspects of a complete interacting quantum field theory is still not always available.

Powerful methods to analyze interacting quantum fields are those provided by the 
renormalization group program first introduced by Wilson \cite{Wilson1, Wilson} and Polchinski \cite{Polchinski:1983gv}. According to this framework the study of features of the interacting field theories is accomplished by introducing suitable regularizations governed by a few parameters and analyzing the way in which properties of the theory change under rescaling of the parameters or analyzing the limits where these regularization parameters vanish.

Motivated by the renormalization group program, 
in this paper we study regularized 
theories which give rise to convergent perturbative expansions of interacting quantum field observables in terms free fields and, in some cases, 
{ we discuss the weak limits in which the regularization is removed and we compare the result with perturbation theory.}
%
%
We are considering uncharged massive ($m>0$) quantum scalar field theories which are self-interacting and which propagate on a $d$-dimensional Minkowski spacetime.
We use perturbation theory to describe the self-interacting scalar field over the free case (outside the adiabatic limit) and we discuss the convergence of the corresponding power series in coupling parameters {when a suitable regularization is employed}.

We shall use the methods of pAQFT to present our results 
\cite{BruDutFre09,RejznerBook, DutschBook}.  
This framework has roots in the works of Brunetti and Fredenhagen \cite{BrunettiFredenhagen00} (see also \cite{DuetschFredenhagen04}) and in the works of Hollands and Wald \cite{HollandsWald2001, HW02, HW05} which combine the analysis of covariant quantum field theories on curved spacetimes \cite{BFV03} with renormalization methods studied by Epstein and Glaser \cite{EpsteinGlaser} and Steinmann \cite{Steinmann} and which are 
used to implement the analysis of Bogoliubov \cite{BogoliubovShirkov} and  Stueckelberg \cite{Stueckelberg,Stueckelberg1}
in perturbation theory. 
In this formalism field observables are described as functionals on field configurations $\phi\in\mathcal{C}$, see \eqref{eq:configurationspace} and the set of functionals $\mathcal{F}$ given in \eqref{def:microcausal}. The quantum features of the analyzed theories are present in the product used to multiply these objects \eqref{eq:prodcut}. Furthermore, with a suitable involution (a subset of) $\mathcal{F}$ has the structure of a $*$-algebra. Once a state $\omega$ is chosen, by the Gelfand Naimark Segal (GNS) construction we can represent them as operators over a suitable Hilbert space. 

The construction of these algebras for linear (free) theories is now well understood \cite{BFV03} and interacting field theories can be obtained with perturbation theory built on top of it.
Although the interacting fields obtained in this way are at best understood as formal power series in the coupling parameters, various implementations of renormalization group ideas can be studied also within that framework \cite{BruDutFre09}.

\bigskip
{\bf Setup of the problem}\\
The perturbative construction of interacting fields is realized by the celebrated Bogoliubov map \cite{BogoliubovShirkov}, which  is  used to represent interacting fields in terms of fields of the free theory.
More specifically, the Bogoliubov map allows us to construct
expectation values of time-ordered products $\mathcal{T}(\phi_1^I, \dots ,\phi_n^I)$  of interacting fields $\phi^I$ in a state $\omega$  as suitable expectation values of objects constructed with elements of the free theory
\[
\langle \mathcal{T}(\phi_1^I, \dots ,\phi_n^I)\rangle_\omega = \langle S^{-1}\mathcal{T}(S, \phi_1,\dots, \phi_n)\rangle_\omega.
\]
Here, the free theory we are working with is the massive Klein-Gordon field, propagating on a $d$-dimensional globally hyperbolic spacetime $M_d$, and its Lagrangian density $\mathcal{L}_0$ is given in \eqref{eq:freeLagrangian}.
Furthermore, the $S$-matrix of the theory, 
\begin{equation}\label{eq:S-matrix-intro}
S = \mathcal{T} \exp (\mathrm{i} L_I) = \sum_{n\geq 0} \frac{\mathrm{i}^n}{n!} \mathcal{T}(L_I,\dots , L_I) 
\end{equation}
is defined as the time-ordered exponential of the interaction Lagrangian 
\begin{equation}\label{eq:LI-intro}
L_I = \lambda\int  \check{f}(\phi(x)) g(x) \dvol x
\end{equation}
where $\lambda$ is the coupling parameter, $g$ is a compactly supported smooth function and $\check{f}$ is a sufficiently regular 
function (usually a polynomial). 
Interaction Lagrangians of this form with $\check{f}$ bounded are also considered in a very interesting recent paper \cite{dybalski}.
The power series defining $S$ is in the generic case only a formal power series.
Actually, the construction of $S$ given above in \eqref{eq:S-matrix-intro}, and thus of the Bogoliubov map, is plagued by problems of various kinds that have so far been only partially solved.
%

The first class of problems we would like to mention is of ultraviolet nature and is well understood in the literature.
In particular, in the functional formalism $L_I$ given in \eqref{eq:LI-intro} is an element of the set of local functionals $ \mathcal{F}_{\text{loc}} \subset \mathcal{F}$ given in \eqref{eq:localfunctionals}. The time-ordered product among $n$ local functionals is constructed with operations similar to the one introduced in  \eqref{eq:star-product} with the Feynman propagator ($\Delta_{F,0}$ given in \eqref{eq:DeltaF}) in place of the two-point function $w$.
The time-ordered product constructed in this way
gives rise to 
 products of distributions (various combinations of the Feynman propagator) which are well defined only on $C^{\infty}_0(M_d^{n}\setminus D)$ where $D$ is the total diagonal in $M_d^{n}$ which is a region of codimension $d$ in $M_d^{n}$.
The naive extension of these distributions to the full spacetime gives rise to divergences. 
These divergences can be cured using certain renormalization prescriptions.
In mathematical terms, a renormalization prescription consists in the procedure of extending  distributions (over the total diagonal $D$) keeping the scaling degree.
Among the various renormalization prescriptions available in the literature, we recall here the recursive procedure of Epstein and Glaser \cite{EpsteinGlaser} and Steinmann \cite{Steinmann} and its recent implementation on curved spacetime \cite{BrunettiFredenhagen00, HollandsWald2001,HW02}.
{ In particular, we refer to \cite{BrunettiFredenhagen00, HollandsWald2001,HW02} and to \cite{Dutsch:2014aa}, for the discussion about the renormalization scheme which can be used on a generic Lorentzian spacetime to construct time-ordered products of polynomial interaction Lagrangians and the corresponding $S$-matrices as formal power series in the coupling constant. We recall that these methods are compatible with the zero momentum method of Bogoliubov-Parasiuk-Hepp-Zimmermann (BPHZ) \cite{BogoPara, Hepp, Zimmermann}.}
%
%

The second class of problems in the construction of the $S$ matrix as in \eqref{eq:S-matrix-intro}, are of infrared nature and they becomes manifest when the limits $g$ to $1$ are taken (adiabatic limits), where $g$ is the cutoff function inserted in the interaction Lagrangian $L_I$ given in eq \eqref{eq:LI-intro}. In the case of Lorentzian theories, these kinds of problems are understood in the literature  with the help of Einstein causality, see e.g.  \cite{BrunettiFredenhagen00, HollandsWald2001}. 

Even after having solved all these issues, the convergence of the power series in $\lambda$ which defines $S$ in \eqref{eq:S-matrix-intro} is in general not available.

In particular, if $\check{f}$ in $L_I$ in \eqref{eq:LI-intro} is a polynomial of degree larger than $2$, and if the spacetime dimensions are larger than $3$, $S$ can be constructed only as a formal power series in the coupling parameters $\lambda$. This formal power series is in the  best case scenario only an asymptotic expansion.

The program of constructive quantum field theory has however constructed the interacting quantum field theories in the case of low-dimensional theories, see the book \cite{GlimmJaffe} and references therein.
Usually there the analysis is conducted in a Euclidean spacetime and the Lorentzian case is obtained by means of a suitable analytic continuation done by means of the Osterwalder Schrader theorem \cite{OsterwalderSchrader}.
In the two-dimensional case this has been accomplished for $\lambda{\phi^4_2}$ in 
\cite{GJ1, GJ2, GJ3}.
Similar analyses for the Sine-Gordon model can be found in 
\cite{Froe76, FS}, see also the recent results
\cite{BFM1, Bahns:2016zqj, Bahns:2021aa,  BahnsPinamontiRejzner}.

In the case of three dimensions, Borel summability of the perturbative series has been established in 
\cite{MagnenSenor}, further results about the three-dimensional case can be found in \cite{GJ3d-2, GJ3d}.
Recent results have been obtained by methods of stochastic quantization by Hairer \cite{Hairer}, 
see also the results of \cite{Gubinelli}.
The combination of these ideas with methods of pAQFT has been established in \cite{DDRZ}.
Recent results about  $\phi^4$ on the sphere in the three-dimensional case have been obtained in \cite{Bailleul:2023wyv}.
Recent results in the subcritical case have been studied in \cite{Duch:2023uhu}.
In a very nice paper, Buchholz and Fredenhagen \cite{Buchholz-Fredenhagen}, see also \cite{BDFK22, BF21}, 
established a minimal set of axioms and properties satisfied by the $S$ matrices of various interaction Lagrangians which follow form Einstein causality and from the requirement of validity of the equation of motion of the theory.  Imposing these requirements, the representation of the set of all possible $S$ matrices exists by abstract arguments.

\bigskip
{\bf Strategy used in this paper}\\
In this paper we {analyze}
theories which are described by interaction Lagrangians of polynomial type (in \eqref{eq:LI-intro}, we choose $\check{f}(\phi)=P(\phi)$ where $P$ is a polynomial in the field) {in a Lorentzian spacetime}.
We furthermore add two suitable regularizations which allow us to obtain an $S$-matrix described by a power series which converges in a suitable sense. Eventually, we discuss {the removal of the regularization by an abstract argument, showing that in the limit one recovers at least ordinary perturbation theory, together with some properties expected to hold for the full interacting fields like a remnant of locality.} 
%

The first regularization which we also employ has already been used in the literature {in the Euclidean case}, and produces essentially a mean field theory (see the nice works of Kopper on this topic
\cite{Kopper:2022kge}) which is obtained 
by replacing the field configuration $\phi \in \mathcal{C}$ (defined in \eqref{eq:configurationspace}) with 
$G_\Lambda * \phi$
in the interaction Lagrangian
where $G_\Lambda$ is a suitable rapidly decreasing function in space, see below in equation \eqref{eq:G-lambda} for its precise form.
This regularization corresponds to adding a suitable cutoff in the propagators of the theory. 
En passant we observe that this regularization arises naturally if one assumes the spacetime to be noncommutative \cite{Doplicher:2019qlb}.

In the literature,
to remove 
this kind of regularization (and similar regularizations used to cure infrared divergence) a recursive procedure is usually employed. 
We recall here the recursive procedure (multiscale analysis)  of Gallavotti \cite{Gallavotti} used to construct interacting fields, see also \cite{Gallavotti1, Gallavotti2}. 
A nice application of this multiscale analysis has been used to get Borel summability of planar $\phi^4_4$ theories \cite{Porta}.
Other renormalization group related works are those of Balaban \cite{Balaban1, Balaban2, Balaban3, Balaban4}, see also Dimock \cite{Dimock}. A review about these and various other methods can be found in the book of Rivasseau \cite{Rivasseau}, see also the interesting review paper \cite{RivasseauRev}.
In this paper we shall follow another path, actually, we shall see in Theorem \ref{eq:main-theorem-intro} below that the $S$-matrix constructed with this regularization 
can be represented by a unitary operator on the Hilbert space of the used free reference theory.
Hence, the removal of the utilized regularizations can be realized, at least by abstract arguments, knowing that the unit ball in the set of linear operators over the mentioned Hilbert space is sequentially complete in the weak-$*$ topology. {Controlling the properties of the limit points remains however an open problem.}

The second regularization we are using tames the large values of the interacting fields. 
Combining the two we have an interaction Lagrangian of the form 
\begin{equation}\label{eq:LI-cutoffs-intro}
L_{I,\Lambda_1, \Lambda_2} = \lambda\int_M e^{-\Lambda_1^2 (G_{\Lambda_2}*\phi(x))^2} P(G_{\Lambda_2}*\phi(x))    g(x)  \dvol x, \qquad g \in C^\infty_0(M)
\end{equation}
where $P$ is a suitable polynomial and $*$ the ordinary convolution. This interaction Lagrangian 
depends on two positive regularization parameters $\Lambda_1$, $\Lambda_2$.
The key observation is that now $L_{I,\Lambda_1, \Lambda_2}$ can be rewritten as
\begin{equation}\label{eq:LI-vertex}
L_I = \lambda \int_M \int_\mathbb{R}  e^{\mathrm{i} a \phi(x)} f(a) g(x) \dvol a \dvol{x}
\end{equation}
where $f$ is a suitable Schwartz function {which depends on $\Lambda_1$ and where the regularization due to $G_{\Lambda_2}$ is taken into account by a suitable modification of the various propagators of the theory.}
The time-ordered products of the $S$ matrix constructed out of this interaction Lagrangian involve only regularized products of vertex operators $V_a(g) = \int_M g(x) e^{\mathrm{i} a \phi(x)} \dvol x$. Furthermore, the bounds satisfied by products of vertex operators do not grow too heavily with the perturbative order $n$ and they can be bounded by $Ce^{n}$. On the contrary, in the case of polynomial interaction of $\lambda\phi^4$ type,
the number of Feynman diagrams used to describe interacting fields at order $n$ for the $\lambda\phi^4$ case grows as $(4n)!/((2n)!2^{2n})$.

\bigskip
{\bf Main results}\\
With this observation the first main result present in this paper is the following Theorem
whose proof is obtained by combining Theorem \ref{th:convergence-Smatrix} of section 
\ref{se:regularzedSmatrix}
and Theorem \ref{th:convergence-unitary1} of section \ref{se:operators-in-hilbertspace}.

\begin{theorem}\label{eq:main-theorem-intro}
In a $d\geq 2$ dimensional Minkowski spacetime $M$, consider the power series in $\lambda$ defining $S(L_I)$, the $S$-matrix constructed with the regularized propagators for an interaction Lagrangian, $L_I$ in \eqref{eq:LI-vertex}
described by {a Schwartz function $\hat{f}$ of the field, such that $f$ has finite weighted $L_1$ norm, namely it satisfies condition \eqref{eq:condition-A}}, and localized in a compact spacetime region.
When tested on a field configuration $\phi$, the power series of $S(L_I)$ is absolutely convergent, uniformly in the field configuration. 
In the GNS representation of the Minkowski vacuum, $(\mathcal{H},\pi,\Psi)$, the power series defining $\pi(S(L_I))$ converges 
{ absolutely in $\mathcal{B}(\mathcal{H})$ to a unitary operator}.
%
\end{theorem}

We then turn to discussing the limit in which the regularization parameters are removed ($\Lambda_i\to 0$ in $S(L_I)$).
As mentioned above, in order to take these limits, we use an abstract argument. We actually exploit the sequential completeness in the weak-$*$ topology of the unit ball of linear operators on $\mathcal{H}$. This unit ball contains all the regularized $S$-matrices
$\pi(S(L_{I,\Lambda_1,\Lambda_2}))$.
We can thus extract convergent subsequences. In this way we however lose control over the uniqueness of the limit, {at the same time we do not know if the operator obtained in the limit is a unitary operator, we nevertheless show in section \ref{se:weak-convergence} that the corresponding interacting fields satisfy some causal properties
expected to hold in a Lorentzian spacetime. Furthermore,}
we 
explicitly check that the selected convergent sequences are in agreement with ordinary results in perturbation theory
{ in the sense that when the very same limits are applied at fixed perturbative order one obtains ordinary estimates in perturbation theory}.
This is done in the case of low spacetime dimensions. In the two-dimensional case we study in section \ref{se:two-d-limits}
modifications of the interaction Lagrangian typical of the Sine-Gordon model. In two dimensions the UV singularities shown in the propagator are particularly mild and there is no need to use any renormalization. 
In the three-dimensional case, we get the second main result of this paper which is summarized in the following {Proposition} 
which descends from 
{ Proposition \ref{th:3d-limits-noLambda} and Proposition \ref{eq:th-*weak-convergence}}
 proved in section \ref{se:limit3d}.

\begin{proposition}
In the three-dimensional case, the regularization used in the propagators and in the deformation of $\phi^4$ to a Schwartz function of $\phi$ ({ with corresponding Fourier transform of bounded weighted $L_1$ norm, condition \eqref{eq:condition-A}}) can be removed, 
taking the limits of vanishing $\Lambda_i$ in $S(L_{I,\Lambda_1,\Lambda_2})$ with $L_{I,\Lambda_1,\Lambda_2}$ in \eqref{eq:LI-cutoffs-intro}, in the appropriate direction, giving rise to convergent sequences of unitary operators (in the weak-$*$ topology) in the GNS representation of the Minkowski vacuum.
The coefficients of the expansion in powers of the coupling constant $\lambda$ {converge to known results of perturbation theory.}
\end{proposition}

The extension of the results about the limits of vanishing $\Lambda_i$ to the four-dimensional Minkowski case is not straightforward essentially because of two problems. 
The first problem is the fact that $\lambda\phi^4_4$ is not super-renormalizable and the renormalized interaction Lagrangian is itself known only as a power series in the coupling constant.
The second problem is in the fact that the renormalization theory in this case requires us to renormalize the kinetic terms. Hence contributions proportional to $\phi\Box\phi$ need to be added to the renormalized interaction Lagrangian. These kinds of terms cannot be easily obtained using only one type of vertex operator as in \eqref{eq:LI-vertex}. {In section \ref{se:limit4d} we present a detailed discussion on how these obstructions could be overcome}.
In particular, the first problem could be solved considering a sequence of interaction Lagrangians labelled by the natural number $N$ obtained truncating the series of the renormalization constants at order $N$. 
The second problem is solved using the wavefunction renormalization as a rescaling of the field and inserting the corresponding cutoffs in a redefined set of renormalization constants. 

The results obtained in this paper require that $g$ be smooth and of compact support in $L_I$. 
Haag's theorem precludes taking  the adiabatic limit, namely the limits where $g$ tends to $1$ and obtaining sensible unitary operators for $S$ in that limit. 

However, adiabatic limits studied for the Bogoliubov map applied to local fields or their products are treated 
in perturbation theory using Einstein causality.
Actually Einstein causality allows us to describe this limit as a well-defined direct limit \cite{BrunettiFredenhagen00, HollandsWald2001, HW02}.
{ Motivated by this observation, we can thus consider the limits where the interaction Lagrangian is integrated over the whole space or the whole time, namely we choose $g$ to be factorized $g(t,x)=\eta(t)h(x)$ where $t$ is a time coordinate and $x$ are spatial coordinates and we can take the limits where either $h$ tends to $1$ or $\eta$ tends to $1$.}
In section \ref{se:adiabatic-limits}
we discuss this problem employing a Mayer expansion of the logarithm of the $S$-matrix and we use Kirkwood-Salsburg equations and properties of their solutions to bound the coefficients of the Mayer expansion. 
The results of this section are obtained by adapting results typical of quantum statistical mechanics e.g. theorems of Penrose and Ruelle \cite{Penrose, Ruelle}, see also the lecture notes \cite{Procacci}. {An alternative equivalent procedure would be to use results obtained by Brydges and Kennedy in \cite{Brydges:1987aa}.}
We finally observe that the results obtained in this paper do not make heavy use of the flatness of the spacetime in which we are working.
For this reason an extension of these results to a generic curved background seems to be straightforward.

\bigskip
{\bf Structure of the paper}\\
The paper is organized as follows.
In Section \ref{se:algebra-propagator} we recall some facts about pAQFT, we introduce the regularized propagators and we discuss some properties of the algebra generated by Vertex operators with regularized products.
Section \ref{se:theorems} contains the theorems about convergence of the $S$-matrices for regularized interaction Lagrangians. 
In particular, in subsection \ref{se:regularzedSmatrix} we discuss the form of the interaction Lagrangian we shall use, we analyze its regularized time-ordered products and we prove weak convergence of the series defining the $S$-matrix.
In subsection \ref{se:operators-in-hilbertspace} we analyze the {absolute convergence of the sum defining $S$ in the operator topology} over the Hilbert space of the vacuum of the theory. {In section \ref{se:weak-convergence} we discuss the removal of the cutoffs in the weak-* topology and we discuss some causal properties satisfied by the interacting fields obtained in the limit.}
In Section \ref{se:limits-no-cutoffs} we discuss the limits where the regularization is removed in two, three and {partially} four spacetime dimensions.
Finally in Section \ref{se:adiabatic-limits} we make some considerations about the adiabatic limits.

\section{Algebras and regularized propagators}
\label{se:algebra-propagator}

In this section we introduce a concrete realization of the field observables we are working with. 
We use the functional formalism first introduced in \cite{BruDutFre09, RejznerBook, DutschBook} to analyze a scalar quantum field theory which propagates on a 
$d$-dimensional Minkowski spacetime $(M,\eta)$ with the signature $(-,+,\dots,+)$ with $d\geq 2$.
The free quantum field theory we start with is a real quantum scalar field theory of mass $m>0$. The Lagrangian density of this theory for a field $\varphi$ on $M$ is
\begin{equation}\label{eq:freeLagrangian}
\mathcal{L}_0 = -\frac{1}{2} \partial_\mu\varphi\partial^\mu\varphi  - \frac{1}{2} m^2 \varphi^2. 
\end{equation}
The observables of both the classical and quantum theory are concretely described by functionals over field configurations. The set of the {\bf field configurations} we are using is denoted by
\begin{equation}\label{eq:configurationspace}
\mathcal{C} :=  C^{\infty}(M,\mathbb{R}),
\end{equation}
where the possible values taken by $\phi$ are assumed to be real because we are considering uncharged fields. 
Notice furthermore that no equations of motion are imposed at this stage.      

The observables we are working with are concretely described as complex-valued functionals over the field configurations $\mathcal{C}$.  
The functionals we are considering
are called {\bf microcausal functionals} and they are defined as belonging to the set
\begin{equation}\label{def:microcausal} 
\mathcal{F}:= \{A:\mathcal{C}\to\mathbb{C}\, |\, 
A^{(n)}(\phi)\in \mathcal{E}'(M^n),
\WF(A^{(n)}) \cap \left(\mathscr{V}_+^{n} \cup
\mathscr{V}_-^{n} 
\right)= \emptyset ,\,
n\in \mathbb{N}, \phi\in \mathcal{C}\}.
\end{equation}
Elements of $\mathcal{F}$ 
are smooth and compactly supported, in the sense that for every $F$ in $\mathcal{F}$, its $n$-th order functional derivative $F^{(n)}$ exists as a compactly supported distribution. 
Furthermore the wave front set of these distributions is restricted to be microcausal.
Actually, in \eqref{def:microcausal},
 $\mathscr{V}_{+}$ and $\mathscr{V}_{-}$ are the subsets of the cotangent space $T^{*}M$ whose elements $(x,k)$ have covectors $k$ which are respectively future and past directed and 
$\mathscr{V}_+^n \subset  (T^{*}M)^n$ denotes its Cartesian product. The condition on the wave front set ensures that it is possible to multiply these objects with products constructed with two-point functions of Hadamard type \cite{HollandsWald2001, HW02, BFK}.

We equip $\mathcal{F}$ with a topology  
constructed in such a way that $F_n$ converges to $F$ if and only if for every $k$, and for every $\varphi\in\mathcal{C}$, $F_n^{(k)}(\varphi)$ converges to $F^{(k)}(\varphi)$ in the topology of $\mathcal{E}'(M^n)$ namely the topology of compactly supported distributions. $\mathcal{F}$ is thus a topological vector space, {see e.g. the book of Rejzner \cite{RejznerBook} for further details}.  

Later, in some cases, we shall also use the {\bf regular functionals} 
\[
\mathcal{F}_{\text{reg}} := \{ A\in \mathcal{F}\, | \, A^{(n)}\in C^\infty_0(M^n)\}
\]
 which are defined as the subset of $\mathcal{F}$ formed by elements whose functional derivatives are compactly supported smooth functions.
Furthermore, {\bf local functionals} are defined as elements of the set
\begin{equation}\label{eq:localfunctionals} 
\mathcal{F}_{\text{loc}} :=
\{A\in \mathcal{F} \,|\, A^{(1)}\in C^{\infty}_0(M),\, \supp A^{(n)}\subset d_n\subset M^n , \, n>1 \} 
\end{equation}
where $d_n:=\{(x,\dots, x)\in M^n |\, x\in M\}$ is the thin diagonal in $M^n$.

\bigskip
Relevant functionals which we shall work with are
the {\bf Wick monomials}
\[
\Phi^n(g)(\phi) := \int \phi(x)^n g(x)  \dvol x, \qquad
g\in C^\infty_0(M), \phi\in \mathcal{C}, 
\]
the {\bf Weyl operators}
\[
W(g)(\phi) := e^{\mathrm{i} \Phi(g)(\phi)}, \qquad
g\in C^\infty_0(M), \phi\in \mathcal{C}, 
\]
the smeared {\bf vertex operators}, 
\[
V_a(g)(\phi) := \int e^{\mathrm{i}a\phi(x)} g(x)  \dvol x, \qquad a\in \mathbb{R},\, g\in C^\infty_0(M),\,\phi\in \mathcal{C},
\]
and their integrals against regular  functions of the parameter $a$.
We want to use these objects as generators of the algebras of quantum observables. 
We thus introduce suitable operations among these objects.
We need in particular to introduce certain products which are 
constructed out of suitable two-point functions $w
\in \mathcal{D}'(M\times M)$ with the requirement that the 
 symmetric part of $w$ is real, and that the antisymmetric part of $w$ is imaginary.
On suitably regular functionals $F,G$ the product we are going to use is
\begin{equation}\label{eq:star-product}
F\star_w G := \mathcal{M} e^{\int_{M^2 } w(x,y)\frac{\delta}{\delta \phi(x)}\otimes \frac{\delta}{\delta \phi(y)}  \dvol x\dvol y} F \otimes G. 
\end{equation}
The operation $\mathcal{M}$ is the pullback on functionals of the map $\phi\in \mathcal{C}\mapsto (\phi,\phi)\in\mathcal{C}^2$,
while $\delta/\delta \phi(x)$ denotes the functional derivative. 
The involution is 
\begin{equation}
    F^{*}(\phi) = \overline{F(\phi)}.
\end{equation}
If $\mathfrak{w}$ is the two-point function of a Hadamard state \cite{Ra96, BFK, KayWald}, 
the product $F\star_{q\mathfrak{w}}G$ applied to any two elements $F,G\in \mathcal{F}$ gives rise to a power series in $q$ whose coefficients are well-defined elements of $\mathcal{F}$ \cite{BruDutFre09}, in general this power series is not convergent.
However, the power series has a finite number of elements if at least one of the factors  $F$ and $G$ is a Wick monomial, furthermore, this power series is convergent if the factors are Weyl operators. As we shall see in section \ref{sse:regularized-product-vertex} the power series in $q$ converges also when the factors are vertex operators and $w$ is at least continuous. 
To avoid these complications we introduce the following algebra which is the {\bf polynomial algebra} of fields and that we shall use to get the GNS representations in suitable states. 
\begin{definition}\label{def:algerbas}
Let $w\in \mathcal{D}'(M\times M)$ be a distribution which has real symmetric part and imaginary antisymmetric part.
We denote by $\mathcal{A}^w$ or by 
$(\mathcal{A}^w,\star_w,*)$ the smallest subset of $\mathcal{F}$ which 
is a $*$-algebra with respect to the product $\star_w$ and the $*$ operation
and which contains the identity and 
$\Phi(g)$ for $g\in C^{\infty}_0(M)$.
\end{definition}

There are other relevant algebras which can be obtained with the product $\star_w$ and with the $*$ operation introduced above. 
In particular, another relevant algebra we shall use later
is the Weyl algebra $\mathcal{W}^w$ which is the smallest subset of $\mathcal{F}$ which contains $W(h)$ with all possible $h$ and it forms an algebra with respect to $\star_w$ and $*$. 
We shall give a precise definition in section \ref{se:operators-in-hilbertspace}.
Finally, when $w$ is at least continuous, we shall use $\mathcal{V}^{w}$ which is the $*$-algebra generated by smeared vertex operators. We shall define and use it in section \ref{sse:regularized-product-vertex}.

\subsection{Regularized propagators and deformed algebras}
In the spirit of the works of Wilson and Polchinski \cite{Wilson, Polchinski:1983gv},
we want to construct an interacting quantum field theory for suitable interaction Lagrangians with the regularized propagators we are going to introduce. 
We are working in a Minkowski spacetime, however, the regularization we shall introduce breaks Lorentz invariance. The regularization  we are considering consists in a suitable cutoff of the high frequencies.  
The {\bf cutoff function} $\chi$ we shall use in this paper is constructed choosing a $\hat{\chi}\in\mathcal{S}(\mathbb{R}^{d-1};\mathbb{R})$  which  is the Fourier transform of $\chi \in C^{\infty}_0(\mathbb{R}^{d-1})$.
The function $\chi$ is smooth and compactly supported, furthermore $\chi$ is positive and symmetric
and $\|\chi\|_1=1$, in this way 
 $\hat{\chi}$ is real and $\hat\chi(0)=1$.
 A similar but not equivalent cutoff ($e^{-\frac{|\mathbf{p}|^2}{2}}$) has been used  in the context of effective quantum field theory on noncommutative spacetime \cite{Doplicher:2019qlb}.
Furthermore, analogous cutoffs of high frequencies have been used extensively in the context of Euclidean quantum field theory see \cite{Kopper:2022kge}.
The {\bf regularized two-point function} we are working with is constructed with an integral kernel which takes the form
\begin{equation}\label{eq:Delta+}
\Delta_{+,\Lambda}(t,\mathbf{x}) = \frac{1}{(2\pi)^{d-1}}\int_{\mathbb{R}^{d-1}}  \frac{1}{2\omega}  e^{\mathrm{i}\mathbf{x}\mathbf{p}} e^{-\mathrm{i}\omega t} 
\hat{\chi}(\Lambda \mathbf{p})^2\,
\dvol^{d-1}\mathbf{p},
\qquad (t,\mathbf{x})\in M, \Lambda >0,
\end{equation}
where $(t,\mathbf{x})$ are standard coordinates of Minkowski space, $\omega=\sqrt{|\mathbf{p}|^2+m^2}$ and where the parameter $\Lambda$  describes the intensity of the regularization. 
For strictly positive $\Lambda$, $\Delta_{+,\Lambda}$ is a smooth function. 
In the limit $\Lambda\to 0$ it converges to a distribution which is the two-point function of the vacuum state of the free massive Klein Gordon field of mass $m$.  

The associated Feynman propagator obtained by time ordering $\Delta_{+,\Lambda}$ is {
\begin{equation}\label{eq:DeltaF}
\Delta_{F,\Lambda}(t,\mathbf{x})  = \frac{1}{(2\pi)^d}\int_{\mathbb{R}^d} \frac{-\mathrm{i}}{-p_0^2 +\mathbf{p}^2 +m^2+\mathrm{i}0^+} 
\hat{\chi}(\Lambda \mathbf{p})^2
e^{\mathrm{i}\mathbf{x}\mathbf{p}} e^{-\mathrm{i}p_0 t}
\,\dvol^dp ,
\qquad (t,\mathbf{x})\in M, \Lambda >0.
\end{equation}
For strictly positive $\Lambda$, $\Delta_{F,\Lambda}$ is a  continuous function,} see e.g. Proposition 4.1 in \cite{Doplicher:2019qlb} for a proof when $e^{-\Lambda^2 |\mathbf{p}|^2}$ is used in place of $\hat{\chi}(\Lambda\mathbf{p})^2$.
{ The limits $\Lambda\to0$ of $\Delta_{F,\Lambda}$ and of $\Delta_{+,\Lambda}$ furnish the ordinary unregularized distributions.}
The functions $\Delta_{+,\Lambda}$ and  $\Delta_{F,\Lambda}$ are used to construct elements of $\mathcal{D}'(M\times M)$ using convolution and the standard pairing. In particular 
\[
\tilde{\Delta}_{+,\Lambda}(h,g) = \langle h, \Delta_{+,\Lambda}*g\rangle = \int_{M^2} h(x) \Delta_{+,\Lambda}(x-y) g(y) \dvol x \dvol y
\]
and similarly for $\Delta_{F,\Lambda}$.  
With a little abuse of notation we shall drop the symbol $\tilde{\,}$ and we shall denote these distributions with the symbol of their integral kernel.
The product $\star_{\Delta_{+,\Lambda}}$ constructed as in \eqref{eq:prodcut} will be denoted by $\star_\Lambda$, while the product obtained in the limit of vanishing $\Lambda$ is $\star_{\Delta_+}$ and we shall denote it simply by $\star$.

Products constructed with the propagators $\Delta_{+,\Lambda}$ and $\Delta_{F,\Lambda}$ can be 
related to products constructed at $\Lambda=0$ by means of the pullback on functionals of the map $\iota_\Lambda$ acting on field configurations we are going to introduce. 
The map $\iota_\Lambda:\mathcal{C}\to\mathcal{C}$ is defined as 
\begin{equation}\label{eq:iota}
\iota_\Lambda(\phi) := G_\Lambda * \phi
\end{equation}
where $*$ denotes the ordinary convolution on $M$ and where the {generalized function} $G_{\Lambda}$ is 
\begin{equation}\label{eq:G-lambda}
G_\Lambda(x_0,\mathbf{x}) := 
\delta(x_0) \chi(\mathbf{x}/\Lambda)/\Lambda^{d-1}
, \qquad 
x_0\in\mathbb{R}, \,\mathbf{x}\in \mathbb{R}^{d-1},
\end{equation}
where $\delta$ denotes the Dirac delta function and $\chi\in C^{\infty}_0(\mathbb{R}^3)$ is the cutoff function chosen at the beginning of this section (it is symmetric, positive and real valued and $\|\chi\|_1=1$). 
The convolution of a compactly supported function with a smooth function is again smooth and the convolution with the Dirac delta is the identity, hence, the action of $\iota_\Lambda$ is closed in $\mathcal{C}$.

The pullback of $\iota_\Lambda$ on functionals is denoted by 
\begin{equation}
\label{eq:rlambdapullback}
r_{\Lambda} F(\phi) := F(\iota_\Lambda \phi), \qquad F\in\mathcal{F}.
\end{equation}
\begin{lemma}
The action of $r_\Lambda$, the pullback of $\iota_\lambda$ on $\mathcal{F}$, is internal in $\mathcal{F}$.
\end{lemma}
\begin{proof}
We observe that, for every $F\in \mathcal{F}$, $r_\Lambda F$ is a well-defined functional, furthermore, it is smooth and of compact support because the cutoff function $\chi$ used in $\iota_\Lambda$ in \eqref{eq:iota} is of compact support.
It is also microcausal because $\mathfrak{g}(h,g) = \langle h, G_{\Lambda}*g\rangle$, 
with $\mathfrak{g}\in \mathcal{D}'(M^2)$ and $h,g\in C^{\infty}_0(M)$,
is such that  
$\WF(\mathfrak{g})\cap (\mathscr{V}_+^2 \cup \mathscr{V}_-^2)  = \emptyset$ hence by an application of Theorem 8.2.13  of \cite{HormanderI} the composition of the distribution $ 
F^{(n)} \in \mathcal{E}'(M^n)$
with $\mathfrak{g}^{\otimes n}$ seen as a distribution $\mathcal{E}'(M^n\otimes M^n)$ has a wavefront set which is microcausal, 
$\WF(F^{(n)}\circ \mathfrak{g}^n )\cap (\mathscr{V}_+^n \cup \mathscr{V}_+^n)  = \emptyset$.
Hence the composition of
 a microcausal functional with 
 $\iota_\Lambda$ is again microcausal.
\end{proof}

The map $r_{\Lambda}:\mathcal{F}\to\mathcal{F}$ can be promoted to a $*$-homomorphism of algebras. 
From now on, we shall denote $\star_{\Delta_{+,\Lambda}}$ by $\star_\Lambda$ and $\star_{\Delta_{+}}$ by $\star$. 
\begin{lemma}\label{le:*-homomorphism}
Consider the restriction to 
$\mathcal{A}^{\Delta_{+,\Lambda}}$
of 
$r_\Lambda$ given in \eqref{eq:rlambdapullback}, which is the pullback to $\mathcal{F}$ of the map $\iota_\Lambda$ acting on $\mathcal{C}$  given in \eqref{eq:iota}. For every 
$F,G\in\mathcal{A}^{\Delta_{+,\Lambda}}$
it holds that
\begin{equation}\label{eq:cond-*-homomorphisms}
    r_\Lambda(F\star_\Lambda G) = r_\Lambda(F)\star r_\Lambda(G), \qquad 
r_\Lambda(F)^*=\overline{r_\Lambda(F)}=r_\Lambda(F^*),
\end{equation}
hence $r_\Lambda$ defines a $*$-homomorphism 
\begin{equation}\label{eq:*-homomorphism}
r_\Lambda (\mathcal{A}^{\Delta_{+,\Lambda}},\star_\Lambda, *) \to (
\mathcal{A}^{\Delta_{+}},\star , * )
\end{equation}
where the product $\star_\Lambda$ is constructed with $\Delta_{+,\Lambda}$
 and $\star$ with $\Delta_+=\Delta_{+,0}$.
\end{lemma}
\begin{proof}
Products of linear fields and the identity
generate the algebras $\mathcal{A}^{\Delta_{+,\Lambda}}$ and $\mathcal{A}^{\Delta_{+}}$.
It is thus sufficient to check that conditions \eqref{eq:cond-*-homomorphisms} hold on products of linear fields. These considerations follow by direct inspection observing that $G_{\Lambda}$ given in equation \eqref{eq:G-lambda} is real and that
\[
\Phi(g)(\iota_\Lambda \phi) = \Phi(g)(G_\Lambda*\phi) = \Phi(G_\Lambda * g)(\phi)
\]
furthermore, $\langle G_{\Lambda} * h, \Delta_+*G_{\Lambda} * g \rangle=\langle h, \Delta_{+,\Lambda} * g \rangle$. 
\end{proof}

\subsection{Regularized product rules among vertex operators}\label{sse:regularized-product-vertex}

For our purposes we want to be able to multiply more singular objects than those present in $\mathcal{F}^p$, the subset of polynomial functionals in $\mathcal{F}$ in \eqref{def:microcausal}. Actually, we need to multiply vertex and Weyl operators. 
Furthermore, we also need to use time-ordered products among these objects. 
The product formula for vertex operators can be obtained when a regularization is taken into account.
We have the following product rule
\begin{equation}\label{eq:prodcut}
e^{\mathrm{i}\Phi(f)}\star_\Lambda e^{\mathrm{i}\Phi(g)} = e^{\mathrm{i}\Phi(f+g)} e^{-\langle f,\Delta_{+,\Lambda} g\rangle}, \qquad f,g\in C^{\infty}_0(M).
\end{equation}
This can be obtained analyzing by direct inspection the convergence of the power series in $q$
which defines the product $\star_{q\Delta_{+,\Lambda}}$
 near $q=1$.

In view of the fact that 
the integral kernel of $\Delta_{+,\Lambda}$ is smooth and that the integral kernel of $\Delta_{F,\Lambda}$ is continuous, the product formula among the exponentials of the smeared linear fields given in equation \eqref{eq:prodcut} and the corresponding product formula constructed with the time-ordered propagator can be extended to vertex operators. 
This gives the following {\bf $\star$-product} among vertex operators
\begin{equation}\label{eq:star-product-vertex}
{\prod_{i=1}^n}^{\star_\Lambda} V_{a_i}(g_i) = \int_{M^n} \prod_{i=1}^n   e^{\mathrm{i} a_i \phi(x_i)} e^{-\sum_{l<j}a_la_j\Delta_{+,\Lambda}(x_l,x_j)  }g_i(x_i) \dvol x^i, \qquad g_i\in C^{\infty}_0(M)  .
\end{equation}
Similarly, for the {\bf time-ordered products}
\begin{equation}
\label{eq:timeordered-product-vertex}
{\prod_{i=1}^n}^{\cdot_{T,\Lambda}} V_{a_i}(g_i) = \int_{M^n} \prod_{i=1}^n   e^{\mathrm{i} a_i \phi(x_i)} e^{-\sum_{l<j}a_la_j\Delta_{F,\Lambda}(x_l,x_j)  }g_i(x_i) \dvol x^i , \qquad g_i\in C^{\infty}_0(M) .
\end{equation}
Similar operations can be obtained substituting $\Delta_{+,\Lambda}$ with different two-point functions. Actually, we shall use {$w=w_s+\frac{\mathrm{i}}{2}\Delta_{\Lambda}$, where $w$ is some} symmetric smooth function 
which is obtained from the two-point function of a Hadamard state
$\Omega\in\mathcal{D}'(M^2)$ \cite{KayWald, Ra96}   
with both entries convoluted with 
$G_{\Lambda}$ given in \eqref{eq:G-lambda}
\[
w(h,g) = \Omega(G_\Lambda(h),G_\Lambda(g)).
\]
$w_s$ is its symmetric part,  $\Delta_\Lambda$ the antisymmetric part which is proportional to the causal propagator of the theory convoluted two times with $G_{\Lambda}$ given in \eqref{eq:G-lambda}.

We analyze now certain bounds satisfied by the products of various vertex operators
where, according to the discussion presented above,  the integral kernel of  $w$ is by construction continuous and its symmetric part $w_s$ is real and positive. 

Notice that $w_F = w_s +\frac{\mathrm{i}}{2} \Delta_\Lambda+\mathrm{i} \Delta_{\Lambda,A} $ and
$w = w_s +\frac{\mathrm{i}}{2} \Delta_\Lambda$.
Hence, 
since $w_s$ is real we have the following
\begin{equation}\label{eq:stima-prod}
\begin{aligned}|V_{a_1}(g_1)\cdot_{T,w}\dots \cdot_{T,w} V_{a_n}(g_n)| 
&
\leq  
\int_{M^n} \prod_i| g_i(x_i)| \dvol x^i e^{- \sum_{l<j} a_l a_j w_s(x_l,x_j)}
\\
|V_{a_1}(g_1)\star_w\dots \star_w V_{a_n}(g_n)| 
&
\leq  
\int_{M^n} \prod_i|g_i(x_i)| \dvol x^i e^{- \sum_{l<j} a_l a_j w_s(x_l,x_j)}.
\end{aligned}
\end{equation}
We furthermore observe that the right-hand sides of the previous inequalities coincide and can be written as the expectation value of a certain commutative product of vertex operators. 
In particular
\begin{equation}\label{eq:commutative-product}
\ev{ V_{a_1}(g_1)\cdot_{w_s}\dots \cdot_{w_s} V_{a_n}(g_n)}_0
:=
\int_{M^n} \prod_ig_i(x_i) \dvol x^i e^{- \sum_{l<j} a_l a_j w_s(x_l,x_j)}
\end{equation}
where $\ev{F}_0=F(0)$, $F\in \mathcal{F}$ denotes the evaluation on the vanishing field configuration and it defines a state on the commutative algebra generated by the identity and the vertex operators with the product $\cdot_{w_s}$. 
The product $\cdot_{w_s}$ is the commutative product constructed as $\star_{w_s}$ in \eqref{eq:prodcut}. 
We use a different symbol to stress that it is commutative. 
    
Actually, here $w_s$ is symmetric, furthermore $w_s$ is real, 
$a_i$ are also real, and $g_i$ are positive hence, using the Cauchy-Schwarz inequality for the state $\ev{\cdot}_0$ gives for $g_i\geq 0$
\begin{align*}
|\ev{V_{a_1}\cdot_{w_s}\dots \cdot_{w_s} V_{a_n}}_0|^2
&\leq
|\ev{V_{a_1}\cdot_{w_s}\dots \cdot_{w_s} V_{a_n}\cdot_{w_s}
V_{-a_1}\cdot_{w_s}\dots \cdot_{w_s} V_{-a_n}}_{0}|
\\
=  
\int_{M^{2n}}  \!\!\!\!\!\!
e^{- \sum\limits_{l<j} a_l a_j w_s(x_l,x_j)}
&e^{- \sum\limits_{l<j} a_l a_j w_s(y_l,y_j)}
e^{\sum\limits_{l,j} a_l a_j w_s(x_l,y_j)}
\left(\prod_{i=1}^n g_i(x_i)g_i(y_i)  \right)
\dvol x \dvol y,
\end{align*}
where here and below, if not explicitly stated, $V_{a_i}=V_{a_i}(g_i)$ for suitable compactly supported smooth test functions.
Since $w_s$ is also positive we can estimate it using a technique similar to inverse conditioning of Fr\"ohlich used in other contexts to get a better estimate \cite{Froe76}. We have actually the following Lemma

\begin{lemma}\label{le:estimate}
Let $w_s$ be a positive, symmetric, translation invariant distribution whose integral kernel is described by a continuous bounded real function on $M^2$. Assume $0\leq g\in C^{\infty}_0(M)$. Consider $a_{i}\in \mathbb{R}$ with $i\in \{1,\dots, n\}$.
Set $W=w_s(x,x)$. The following estimate holds uniformly on the field configurations $\phi \in \mathcal{C}$.
\[
|V_{a_1}(g_1)\cdot_{w_s}\dots \cdot_{w_s} V_{a_n}(g_n)|
\leq  
e^{ \sum\limits _{i } a_i^2 \frac{W}{2}} \prod_{i=1}^n \|g_i\|_1, \qquad g_i\in C^{\infty}_0(M;\mathbb{R}^{+}).
\]
\end{lemma}
\begin{proof}
Let 
$\Psi(z) = \sum_{i=1}^n \left( a_i\delta(z-x^i) - a_i\delta(z-y^i)  \right) $ (where $\delta$ is the Dirac delta function) be a distribution of compact support on $M$.
By direct inspection we have that 
\[
|V_{a_1}\cdot_{w_s}\dots \cdot_{w_s} V_{a_n}|^2 
\leq
\int_{M^{2n}}  \!\!\!\!\!\!
e^{ -\frac{1}{2} w_s(\Psi,\Psi)  }
e^{ +\frac{1}{2} \sum_i a_i^2  w_s(x^i,x^i) +\frac{1}{2} \sum_i a_i^2  w_s(y^i,y^i)}
\left(\prod_{i=1}^n g_i(x_i)g_i(y_i) \dvol x^i\dvol y^i \right)
\]
Since $w_s$ is positive and $\Psi$ is real valued we have 
\[
|V_{a_1}\cdot_{w_s}\dots \cdot_{w_s} V_{a_n}|^2  
\leq  
\int_{M^{2n}}  \!\!\!\!\!\!
e^{ \frac{1}{2} \sum_i a_i^2  w_s(x^i,x^i) +\frac{1}{2} \sum_i a_i^2  w_s(y^i,y^i)}
\left(\prod_{i=1}^n g_i(x_i)g_i(y_i) \dvol x^i\dvol y^i \right).
\]
Finally since $w_s$ is invariant under translations and its integral kernel is continuous we have $w_s(x^i,x^i)=W$ for every $i$, hence
\[
|V_{a_1}\cdot_{w_s}\dots \cdot_{w_s} V_{a_n}|^2  
\leq  
e^{  \sum_{i=1}^n a_i^2 W}
\int_{M^{2n}}  \prod_{i=1}^n g_i(x_i)g_i(y_i) \dvol x^i\dvol y^i 
\leq  
e^{  \sum_{i=1}^n a_i^2 W} \prod_{i=1}^n\|g_i\|_1^{2}
\]
\end{proof}
Using Lemma \ref{le:estimate}, the inequality given in \eqref{eq:stima-prod} and recalling the definition of the commutative product \eqref{eq:commutative-product}, if $|a_i|\leq a$ for every $i$, and if $g_i=g$ for every $i$, we can now estimate
\begin{align*}
|V_{a_1}(g)\cdot_{T,w}\dots \cdot_{T,w} V_{a_n}(g)| 
&\leq  
e^{n\frac{a^2 W}{2}} \|g\|_1^{n}
\\
|V_{a_1}(g)\star_w\dots \star_w V_{a_n}(g)| 
&\leq  
e^{n\frac{a^2 W}{2}} \|g\|_1^{n}.
\end{align*}

With these estimates at our disposal, we can now consider the following 
\begin{definition}\label{def:Vw}
    Consider $w\in \mathcal{D}'(M^2)$ whose integral kernel is continuous and such that $w_s = \text{Re}(w)$ is translation invariant, symmetric and it defines a positive product. 
    We denote by $\mathcal{V}^{w}$ or by $(\mathcal{V}^{w}, \star_w , *)$ the $*$-algebra which contains all possible products of ${\prod_{i=1}^n}^{\star_{w}}\int V_{a}(g_i) f_i(a) \dvol a $ with $g_i\in C^\infty_0(M)$ and $f_i\in \mathcal{S}$ with the property 
    \[
    \left|\int_\mathbb{R} |f_i(a)|e^{\frac{a^2 W}{2}}\right| < \infty.
    \]
    where $W=w_s(x,x)$.
\end{definition}
The products of various $\int V_{a}(g) f(a) \dvol a$ with the requirements of Definition \ref{def:Vw} are well defined thanks to the estimates given in Lemma \ref{le:estimate}.

Consider $w_s=\text{Re} \Delta_{+,\Lambda}$ and $w_F= \Delta_{F,\Lambda}$, we have that both 
$\mathcal{V}^{w_s} = (\mathcal{V}^{w_s},\cdot_{w_s}, *)$ and $\mathcal{V}^{w_F} = (\mathcal{V}^{w_F},\cdot_{T,w}, *)$ are commutative $*$-algebras.
Furthermore, the evaluation on the vanishing field configuration used above
\[
\ev{F}_0 = F(0), \qquad F\in\mathcal{F} 
\]
defines a state both on the commutative algebra $\mathcal{V}^{w_s}$ and on $\mathcal{V}^{w}$ if $w$ is also positive.

\section{Interaction Lagrangians, S-matrices and convergence }\label{se:theorems}

\subsection{The set of interaction Lagrangians and their S-matrices}\label{se:regularzedSmatrix}

We are interested in considering an interacting quantum field theory.
The smeared interaction Lagrangian is 
\[
L_I=\int_M \mathcal{L}_I(x) g(x)\dvol x \in \mathcal{F}
\]
where $g\in C^\infty_0(M;\mathbb{R})$ is an adiabatic cutoff which is inserted to make the interaction Lagrangian of compact support  and where $\mathcal{L}_I$ is  written in terms of the field $\phi$.
We discuss how to regularize this theory to get convergent $S$ matrices which approximate the sought $S$ matrix.

We are interested in considering the case in which 
\[
\mathcal{L}_I(x) :=  \hat f(\phi(x))
\]
where $\hat f\in \mathcal{S}(\mathbb{R};\mathbb{R})$ is a real-valued Schwartz function, as real-valuedness is necessary to ensure formal self-adjointness of $L_I$. 

Later on we shall discuss some limits in which $\hat f$ becomes a polynomial. In this paper we shall not consider the case in which $\mathcal{L}_I$ contains also products of derivatives of the fields. 

Considering $f$, the inverse Fourier transform of $\hat f$, we obtain that $L_I$ can be expressed in terms of the vertex operators in the following way
\begin{equation}\label{eq:interaction-Lagrangian}
L_I(f,g) :=  \int_\mathbb{R}  V_a(g) f(a)  \dvol a
= 
\int_M\!  \int_\mathbb{R}\!\!  e^{\mathrm{i}a\phi(x)} f(a) g(x) \dvol a\dvol x
, 
\quad f\in \mathcal{S}(\mathbb{R}),\, \overline{f(a)}=f(-a),\, g\in C^{\infty}_0(M).
\end{equation}
Since $f$ is a Schwartz function and $g$ a compactly supported smooth function we have that ${L}_I \in \mathcal{F}_{\text{loc}}\subset \mathcal{F}$.
Furthermore, $g$ is real valued and 
 $f$ 
 is chosen in such a way that $L_I(f,g)$ is formally self-adjoint hence ($\overline{f}(a)=f(-a)$).
The function $f$ is furthermore assumed to decay sufficiently rapidly to have that $L_I(f,g)\in \mathcal{V}^{w}$ according to definition \ref{def:Vw} where $w$ is the regularized two-point function of the theory under investigation. 

With this choice, quantum products and/or time-ordered products of $L_I$ 
can now be obtained using the product rules
\eqref{eq:star-product-vertex}
and 
\eqref{eq:timeordered-product-vertex}
given in terms of Vertex operators and their estimates follow from Lemma \ref{le:estimate}.

We start introducing the $S$ matrix of the theory, the time-ordered exponential where the time ordering is built with $w_F$.
\begin{equation}\label{eq:SV}
\begin{aligned}
S_w(L_I) &:= \sum_{n\geq 0}S_{w,n}(L_I) := \sum_{n\geq 0} \frac{\mathrm{i}^n}{n!} 
\mathcal{T}(\underbrace{L_I\otimes \dots \otimes L_I}_{n})
:=\sum_{n\geq 0} \frac{\mathrm{i}^n}{n!} 
L_I^{\cdot_{T,w} n}
\\
&= \sum_{n\geq 0} \frac{\mathrm{i}^n}{n!} 
\int_{\mathbb{R}^n} \dvol a \;
{\prod_i^n}^{\cdot_{T,w}}
f(a_i) V_{a_i}(g)
\end{aligned}
\end{equation}
where $S_{w,n}$ is the $n$-th order contribution in $L_I$ of $S_w(L_I)$ and 
where $\mathcal{T}$ is the time ordering map implicitly defined by the previous relation.

We observe that with our choices
$L_I(f,g)$ is assumed to be an element of  $\mathcal{V}^{w_F}$ and thus the power series in $\lambda$ which defines $S_w$
has coefficients in $\mathcal{V}^{w_F}$.
We have the following lemma about convergence of $S_w(L_I)$.

\begin{theorem}\label{th:convergence-Smatrix}
Consider the symmetric part of the two-point function $w_s$ invariant under translation and such that $w_s(x,x)=W$. Consider  $g\in C^{\infty}_0(M)$, $g\geq 0$, and  $f\in\mathcal{S}'(\mathbb{R})\cap{C^\infty}(\mathbb{R})$ such that $\overline{f}(a)=f(-a)$ and
\begin{equation}\label{eq:condition-A}
\int_{\mathbb{R}} |f(a)|e^{\frac{a^2}{2} W} \dvol a \leq A < \infty
\end{equation}
for some constant $A$.
The sum which defines the S-matrix constructed with the interaction Lagrangian $L_I(f,g)$ given in \eqref{eq:interaction-Lagrangian}  is
\[
S_w(L_I(f,g)) = \sum_{n\geq 0} \frac{\mathrm{i}^n}{n!} \underbrace{L_I(f,g) \cdot_{T,w}\dots \cdot_{T,w} L_I(f,g)}_{n}
\]
and when tested on a generic $\phi\in \mathcal{C}$ it converges absolutely.
Furthermore, the following estimate holds
\begin{align*}
|S_w(L_I(f,g))(\phi)| 
\leq \exp\left({ \|g\|_1\int_{\mathbb{R}} |f(a)|e^{\frac{a^2}{2} W} \dvol a }\right)
= \exp\left({ A \|g\|_1 }\right)
\end{align*}
hence, the sum defining $S_w$ converges uniformly on $\phi$.
\end{theorem}
\begin{proof}
We observe that from the definition given in \eqref{eq:SV}
by direct inspection
\begin{align*}
|S_w(L_I(f,g))| &\leq  \sum_{n\geq 0} \frac{1}{n!} |\underbrace{L_I(f,g) \cdot_{T,w}\dots \cdot_{T,w} L_I(f,g)}_{n}|
\\
&\leq  \sum_{n\geq 0} \frac{1}{n!} \int_{\mathbb{R}^n}|f(a_1)\dots f(a_n)| |{V_{a_1} \cdot_{T,w}\dots \cdot_{T,w} V_{a_n}}| \dvol a_1\dots \dvol a_n
\end{align*}
where the smearing of the vertex operators with respect to $g$ is kept implicit, namely, here and below in this proof $V_a=V_a(g)$.
Notice that $|{V_{a_1} \cdot_{T,w}\dots \cdot_{T,w} V_{a_n}}(\phi)|\leq 
|\ev{V_{a_1} \cdot_{w_s}\dots \cdot_{w_s} V_{a_n}}_0|$ and
using Lemma \ref{le:estimate} to bound the latter, we get
\begin{align*}
|S_w(L_I(f,g))| 
&\leq  \sum_{n\geq 0} \frac{1}{n!} \int_{\mathbb{R}^n}|f(a_1)\dots f(a_n)| 
e^{ \sum_{i=1 }^n \frac{a_i^2 W}{2}} \|g\|_1^{n}
\dvol a_1\dots \dvol a_n.
\end{align*}
The right-hand side does not depend on the field configuration, hence this estimate is uniform in $\phi$.
Furthermore, since  
\[
\int_{\mathbb{R}} |f(a)|e^{\frac{a^2}{2} W} \dvol a \leq A < \infty
\]
we have that the sum defining $S_w$ is uniformly convergent and
\begin{align*}
|S_w(L_I(f,g))| 
&\leq  \sum_{n\geq 0} \frac{1}{n!} |A|^n \|g\|_1^{n} \\
&\leq  e^{ |A| \|g\|_1} 
= \exp\left({ \|g\|_1\int_{\mathbb{R}} |f(a)|e^{\frac{a^2}{2} W} \dvol a }\right).
\end{align*}
\end{proof}
\smallskip
{\bf Remark} 
The factor $e^{\frac{a^2}{2}W }$ is present in the estimates of Lemma \ref{le:estimate} because the fields we are considering are normal ordered with respect to $w_s$.

\smallskip
{\bf Remark} 
We observe that $S_\Lambda$ satisfies some remnant of the standard properties satisfied by the $S$-matrix usually proved in perturbation theory. 
In particular, $S_\Lambda$ is unitary, actually $S_{\Lambda}^{-1}=S_{\Lambda}^{*}$ (for a proof see Proposition 4.4 in \cite{Doplicher:2019qlb}).
The causal factorization property holds along the time direction, see Proposition 4.3 in
\cite{Doplicher:2019qlb}.

\smallskip
Having established weak convergence of the power series which defines the used $S$ matrices we introduce the following notation. In the subsequent part of the paper we shall write $S_{\Lambda}=S_{\Delta_{+,\Lambda}}$ and $S=S_{\Delta_{+}}$.

We conclude this section observing that with  techniques similar to those discussed in the proof of Theorem \ref{th:convergence-Smatrix} we can analyze the convergence of the power series which defines   
\[
S_w(L_I(f,g) + J) 
\]
where $J=\int \phi j \dvol x = \Phi(j)$ with $j\in C^{\infty}_0(M)$. We get that weak convergence of the series (convergence of the series of $S_w$ when evaluated on a generic field configuration) holds uniformly in $\mathcal{\phi}$.
With this at our disposal we can obtain the generating functions of the interacting correlation functions as the relative $S$-matrix
$S_w(L_I)^{-1}\star_w S_w(L_I+J)$.

%

\subsection{Representation as unitary operators on a suitable Hilbert space}
\label{se:operators-in-hilbertspace}

In this section we discuss how to turn the estimates of the weak convergence in the algebra of functionals of the regularized $S$ matrices into a convergence argument in the operator topology to a unitary operator in some representation of the ordinary Weyl algebra constructed over the associated linear theory hence in the GNS representation of a  suitable state.

We start our analysis by recalling a few facts about the Weyl $C^*$-algebra.
The Weyl $C^*$-algebra $\mathcal{W}$ we are referring to is obtained starting from a set of abstract generators called Weyl operators labeled by suitable elements of a symplectic space.
More precisely, the generic Weyl operator is denoted by $\hat{W}(\tilde{f})$,
where $\tilde{f}$ is an element of the symplectic space $C^{\infty}_0(M;\mathbb{R})/PC^{\infty}_0(M;\mathbb{R})$ 
where $P$ is the differential operator which defines the Klein-Gordon equation and whose symplectic form is 
\[
\sigma([\tilde{f}],[\tilde{g}]) = \Delta(\tilde{f},\tilde{g}), \qquad \tilde{f},\tilde{g}\in C^\infty_0(M).
\] 
Weyl generators satisfy the following relations
\begin{equation}\label{eq:Weyl}
\hat{W}(\tilde{f}) \hat{W}(\tilde{g}) = e^{-\frac{\mathrm{i}}{2}\Delta(\tilde{f},\tilde{g})}\hat{W}(\tilde{f}+\tilde{g}), \qquad  \hat{W}(\tilde{f})^{*}=\hat{W}(-\tilde{f}).
\end{equation}
The Weyl $*$-algebra  $\mathcal{W}_0$ 
is then generated by all possible sums and products of $\hat{W}(\tilde{f})$ for various $\tilde{f}$.
This algebra is generated by the identity and $\hat{W}(f)$ with $f$ in $\mathcal{S}$ and where the product is $\star_{\Delta}$.

There exists an unique Weyl $C^*$-algebra $\mathcal{W}$ up to equivalent unitary representation which implements these relations, see e.g. 
\cite{Moretti} and references therein.
The procedure works as follows. We equip $\mathcal{W}_0$ with a suitable $C^*$ norm. This $C^*$ norm is unique up to equivalences. The closure of $\mathcal{W}_0$ with respect to the topology induced by this norm produces the sought  $C^*$ algebra $\mathcal{W}$.

\bigskip
We discuss now how to concretely represent $\mathcal{W}_0$ as an algebra equipped with the product $\star_{\Delta_+}$.
The Weyl operator in $(\mathcal{W}_0,\star_{\Delta_+},*)$ is
\begin{equation}\label{eq:W+renormalized}
W_{+}(\tilde{f})=e^{\mathrm{i}\Phi(\tilde{f})}e^{-\frac{1}{2}\Delta_+(\tilde{f},\tilde{f})}=e^{\mathrm{i}\Phi(\tilde{f})}e^{-\frac{1}{2}\Delta_s(\tilde{f},\tilde{f})}
=e_{\star}^{\mathrm{i}\Phi(\tilde{f})}
\end{equation}
where $\Delta_s=\Delta_+
 -\frac{\mathrm{i}}{2}\Delta$ is the symmetric part of $\Delta_+$. The Weyl relations are satisfied. 

In order to concretely represent elements of $\mathcal{W}$ as operators on a suitable Hilbert space, we select now the state $\omega$ and we consider its GNS representation. The state we are choosing is the Minkowski vacuum, namely the Gaussian state whose two-point function is $\Delta_+$, in particular, on $\mathcal{W}_0$
\begin{equation}\label{eq:stateW+}
    \omega (W_+(f)) := e^{-\frac{1}{2} \Delta_+(\tilde{f},\tilde{f})}=e^{-\frac{1}{2} \Delta_s(\tilde{f},\tilde{f})}.
\end{equation}
This state coincides with the evaluation on the vanishing configuration $\omega(W_+(\tilde{f})) = \ev{W_+(\tilde{f})}_0$

\bigskip
The Weyl $*$-algebra with the deformed  product is now denoted by $\mathcal{W}_\Lambda=(\mathcal{W}_\Lambda,\star_\Lambda,*)$ and the Weyl generator is now 
\[
W_{\Lambda}(\tilde{f})=e^{\mathrm{i}\Phi(\tilde{f})}e^{-\frac{1}{2}\Delta_{+,\Lambda}(\tilde{f},\tilde{f})}
=e_{\star_\Lambda}^{\mathrm{i}\Phi(\tilde{f})},
\]
it satisfies the Weyl relations
\[
W_{\Lambda}(\tilde{f}) 
\star_\Lambda
W_{\Lambda}(g)
=
e^{\frac{1}{2}
\left(\Delta_{+,\Lambda}(\tilde{f},g)-\Delta_{+,\Lambda}(g,\tilde{f})\right)} W_{\Lambda}(\tilde{f}+g) 
\]
and
\[
W_{\Lambda}(\tilde{f})^{*} 
= 
W_{\Lambda}(-\tilde{f})
\]
where $\tilde{f}$ is real.
Observe that with the $*$-homomorphism $r_\Lambda$ introduced above in equation \eqref{eq:*-homomorphism} we can represent in $\mathcal{W}_0$ the Weyl $*$-algebra generated by the identity and by $W_\Lambda(\tilde{f})$ with $\tilde{f}\in\mathcal{S}$, constructed with the product $\star_{\Delta_{+,\Lambda}}$ and which we denote by $\mathcal{W}_\Lambda$.
Furthermore, below we shall use $\omega_{\Lambda}$
which is actually a state characterized by 
\begin{equation}\label{eq:stateWLambda}
\omega_{\Lambda} (W_{\Lambda}(h))=e^{-\frac{1}{2}\Delta_{+,\Lambda}(h,h)}.
\end{equation}
This state on $\mathcal{W}_\Lambda$ 
corresponds to the evaluation on the vanishing configuration.

The following lemma holds

\begin{lemma}\label{le:pullback}
Let $m$ be strictly positive.
Consider the map $r_\Lambda$ introduced in equation \eqref{eq:*-homomorphism}. Its action on the Weyl generators $W_\Lambda$ is
\[
r_{\Lambda }W_{\Lambda}(\tilde{f})= W_{+}(G_\Lambda *\tilde{f}).
\]
The map $r_{\Lambda}:\mathcal{W}_\Lambda\to \mathcal{W}_0$ is a $*$-homomorphism of Weyl algebras.
Consider the Minkowski vacuum $\omega$, which is the Gaussian state on $\mathcal{W}_0$ 
whose action on the Weyl generators is given in 
\eqref{eq:stateW+}
and its GNS triple $(\mathcal{H}_\omega,\pi_\omega,\Psi_\omega)$. It holds that $\pi_{\omega}(r_\Lambda (\mathcal{W}_\Lambda)) \Psi_\omega$ is dense in $\mathcal{H}_\omega$.
\end{lemma}

\begin{proof}
    The first part of this lemma follows from direct inspection. 
    The fact that $r_\Lambda$ is a $*$-homomorphism can be proved as in Lemma \ref{le:*-homomorphism}. 
The state $\omega$ is Gaussian and pure and $\mathcal{H}_{\omega}$ is a Fock space constructed over the one-particle Hilbert space $\mathfrak{h}$. 
The representation $\pi_\omega$ is regular, to prove that 
$\pi_\omega (r_{\Lambda} (\mathcal{W}_\Lambda)) \Psi_\omega$ is dense in $\mathcal{H}_\omega$ it is thus sufficient to prove that 
$D=\{ \Delta_+(G_{\Lambda}*(f_1 + \mathrm{i}f_2))\,|\, f_i\in \mathcal{S}(M;\mathbb{R})\}$ is dense in $\mathfrak{h}$.
In the case of a $d$-dimensional Minkowski space, 
$\mathfrak{h}$ is isomorphic to $L^2(\mathbb{R}^{d-1}
)$.  
Let $\varphi:\mathfrak{h}\to L^2(
\mathbb{R}^{d-1})$ be the map which realizes this isomorphism, 
we have that 
\[
\varphi(D) = \left\{\left.\frac{\hat{G}(p)}{\sqrt{2w_p}}\hat{f}(w_p,p) \,\right|\, f\in \mathcal{S}(M;\mathbb{C})\right\}
\]
where $w_p=\sqrt{p^2+m^2}$.
Notice that the set 
$\varphi(D)$ is dense in $L^2(\mathbb{R}^{d-1})$
because compactly supported smooth functions over $\mathbb{R}^{d-1}$ are contained in 
$\varphi(D)$ and  are 
dense both in $L^2(\mathbb{R}^{d-1})$ and in 
$\varphi(D)$ w.r.t. the standard $L^2$ norm. Since $\varphi$ is an isomorphism of Hilbert spaces, we have that $D$ is dense in $\mathfrak{h}$ and this concludes the proof.
\end{proof}

We have now the following lemma which relates $\omega$ and $\omega_\Lambda$

\begin{lemma}\label{le:pullback-state}
Consider the Minkowski vacuum, which is the state $\omega$ on $\mathcal{W}_0$ 
characterized by equation 
 \eqref{eq:stateW+}.
 The pullback of $\omega$ under $r_\Lambda$ coincides with the state $\omega_\Lambda$ on $\mathcal{W}_\Lambda$ which is the state characterized by \eqref{eq:stateWLambda}.
\end{lemma}
\begin{proof}
    The proof of this statement follows directly from Lemma \ref{le:pullback}.
\end{proof}

The map $r_\Lambda$ given in \eqref{eq:rlambdapullback} intertwines the products $\star$ and $\star_\Lambda$, we notice that it also intertwines the corresponding time-ordered products because the time ordering map commutes with the convolution with $G_{\Lambda}$ given in \eqref{eq:G-lambda}.
We actually have that 
\begin{equation}\label{eq:time-ordering-equivalent}
\begin{aligned}
\Delta_{F,\Lambda}(t,\mathbf{x})
&=
\theta(t)\Delta_{+,\Lambda}(t,\mathbf{x})
+
\theta(-t)\Delta_{+,\Lambda}(-t,\mathbf{x})
\\
&=
\theta(t)G_\Lambda*\Delta_{+}*G_\Lambda
(t,\mathbf{x})
+
\theta(-t)G_\Lambda*\Delta_{+}*
G_\Lambda(-t,\mathbf{x})
\\
&=
G_\Lambda* \Delta_{F}*G_\Lambda(t,\mathbf{x}).
\end{aligned}
\end{equation}
In the last step we used the fact that $G_\Lambda$ is nothing but the Dirac delta function in the time direction. 
With this observation, we can now represent the $S$-matrix with
\begin{equation}\label{eq:repS}
{r_\Lambda}(S_\Lambda(L_I)) = S({r_\Lambda}{L_I})
\end{equation} 
where $S=S_{\Delta_{+}}$ and $S_\Lambda=S_{\Delta_{+,\Lambda}}$.
This observation is useful in the analysis of 
the corresponding operations in the GNS representation of $\omega$. 


We can now prove that the series defining $S({r_\Lambda}{L_I})$ in the GNS representation of the vacuum is absolutely convergent in $\mathcal{B}(\mathcal{H}_{\omega})$ to a unitary operator.
\begin{theorem}\label{th:convergence-unitary1}
Let $\Lambda$ be a strictly positive constant. Let $g$ be a  compactly supported smooth function on $M$ and let $f$ be a compactly supported smooth function on $\mathbb{R}$.
Consider the smeared interaction
 Lagrangian $L_I(f,g)$ of the form given in \eqref{eq:interaction-Lagrangian} with $g\in C^{\infty}_0(M)$, $f \in \mathcal{S}'(\mathbb{R})$ with $\overline{f(a)}=f(-a)$ and such that
\[
\int_{\mathbb{R}} |f(a)|e^{\frac{a^2}{2} W} \dvol a  < \infty\qquad W=\Delta_{+,\Lambda}(x,x).
\] 
Let $\omega$ be the vacuum state of the massive free scalar field on the Weyl $C^*$-algebra $\mathcal{W}$ and let $(\mathcal{H}_\omega, \pi_\omega, \Psi_\omega)$ be its GNS triple. 
The series defining $S_{\Lambda}(L_I)$ converges absolutely in $\mathcal{B}(
\mathcal{H}_\omega
)$ to a unitary operator on $\mathcal{H}_\omega$.
\end{theorem}
\begin{proof}
We analyze
$
\pi_{\omega }(\sum_n  S_n({r_\Lambda}{L_I}))$
%
and we start observing that 
\[
\pi_\omega r_\Lambda L_I (f,g) = 
\int_M\!  \int_\mathbb{R}\!\!  \pi_\omega e^{\mathrm{i}a\phi * G_\Lambda(x))} f(a) g(x) \dvol a\dvol x
= \int \mathcal{V}(t) \dvol t 
\]
where $G_\Lambda$ is given in \eqref{eq:G-lambda} and where  the operator
\[
\mathcal{V}(t) := \int_{\mathbb{R}^{d-1} }\int_\mathbb{R}\!\! \pi_\omega e^{\mathrm{i}a\phi * G_\Lambda(t,\mathbf{x}))} f(a) g(t,\mathbf{x}) \dvol a\dvol \mathbf{x}
\]
is an element of $\mathcal{B}(
\mathcal{H}_\omega
)$. To prove this, we observe that for every $x=(t_x,\mathbf{x})\in M$, $\mathcal{G}_x(y):=G_\Lambda(x-y)$ with $y=(t_y,\mathbf{y})$ is a smooth function of $\mathbf{y}$ multiplied by a delta function supported in $t_x-t_y$, furthermore composing with the causal propagator
$\Delta (\mathcal{G}_x)(y)$ we get a smooth function on $M$. We can thus expand the operator 
\begin{equation}\label{eq:decom-W}
\pi_\omega e^{\mathrm{i}a\Phi(\mathcal{G}_x)} = 
e^{\frac{a^2W}{2}}\hat{W}(a\mathcal{G}_x)
\end{equation}
where here $\hat{W}(a\mathcal{G}_x)$ is the Weyl operator in the GNS representation of the Minkowski vacuum (see \eqref{eq:Weyl} for the Weyl relations) smeared with $a\mathcal{G}_x$. The corresponding representation of this Weyl operator in the $*$-algebra with deformed product $\star_{\Delta_+}$ coincides with $W_+(a\mathcal{G}_x)$ given in
 \eqref{eq:W+renormalized} and thus $\pi_\omega W_+(a\mathcal{G}_x) = \hat{W}(a\mathcal{G}_x)$. Finally, from \eqref{eq:W+renormalized} we get \eqref{eq:decom-W}. 
Hence the following operator-norm bound holds: 
\[
\| \pi_\omega e^{\mathrm{i}a\Phi(\mathcal{G}_x)} \| \leq e^{\frac{a^2W}{2}}\| \hat{W}(a\mathcal{G}_x) \| =e^{\frac{a^2W}{2}}
\]
and similarly
\[
\|\mathcal{V}(t)\| \leq  \int_{\mathbb{R}^{d-1} }\int_\mathbb{R}\!\! e^{\frac{a^2W}{2}} |f(a)| |g(t,\mathbf{x})| \dvol a\dvol \mathbf{x}.
\]
The time-ordered exponential of the operator $\int \mathcal{V}(t)\dvol t$ can be computed and, in view of \eqref{eq:time-ordering-equivalent} it coincides with 
\begin{align*}
S(r_\Lambda L_I(f,g)) 
&=\mathcal T\exp\!\left(\mathrm{i}\int \mathcal{V}(t)\,dt\right)
\\
&= 1+\sum_{n=1}^\infty \mathrm{i} ^n\int_{t_1>t_2>\dots >t_n}   \mathcal{V}(t_1)\dots\mathcal{V}(t_n) \dvol^n t. 
\end{align*}
The series converges in the operator norm of $\mathcal{B}(\mathcal{H}_\omega)$, actually
\begin{align*}
\|S(r_\Lambda L_I(f,g))\| &\leq 1+\sum_{n=1}^\infty \int_{t_1>t_2>\dots >t_n} 
\|\mathcal{V}(t_1)\|\dots \|\mathcal{V}(t_n)\|
%
%
\dvol^n t 
\\
&\leq 1+\sum_{n=1}^\infty  \frac{1}{n!} \|g\|_1^n   \left( \int_{\mathbb{R}} |f(a)| e^{\frac{a^2W}{2}} \dvol a\right)^n    
\end{align*}
where the factor $1/n!$ arises from estimating the volume of the simplex over which the time integrals are taken.
The series is thus absolutely convergent in $\mathcal{B}(\mathcal{H}_\omega)$. Furthermore, $\pi_\omega r_\Lambda L_I(f,g)$ is the time integral of a self-adjoint operator $\mathcal{V}(t)$ and thus the time-ordered exponential of $\mathrm{i}\int \mathcal{V}(t)\dvol t$ which is equal to  
$\pi_\omega S(r_\Lambda L_I(f,g))$ is unitary. 
\end{proof}

We conclude this section with a Corollary of Theorem \ref{th:convergence-unitary1}. Although the following result follows from the Theorem observing that uniform convergence implies also convergence in the strong operator topology we give an independent proof making use of some ideas taken from \cite{Bahns:2021aa}.

\begin{corollary}\label{le:convergence-unitary}
\color{black}
Let $\Lambda$ be a strictly positive constant. Let $g$ be a  compactly supported smooth function on $M$ and let $f$ be a compactly supported smooth function on $\mathbb{R}$.
Consider the smeared interaction
 Lagrangian $L_I(f,g)$ of the form given in \eqref{eq:interaction-Lagrangian} with $g\in C^{\infty}_0(M)$ and with $f \in \mathcal{S}'(\mathbb{R})$ with $\overline{f(a)}=f(-a)$ 
 \[
 \int |f(a)|e^{\frac{a^2 W}{2}}\dvol a < \infty
 \]
where $W = \Delta_{_\Lambda}(x,x)$.
Let $\omega$ be the vacuum state of the massive free scalar field and let $(\mathcal{H}_\omega, \pi_\omega, \Psi_\omega)$ be its GNS triple. 
The series defining $S_{\Lambda}(L_I)$ is a strongly convergent series in $\mathcal{B}(
\mathcal{H}_\omega
)$.
\end{corollary}
\color{black}
\begin{proof}
To prove that the sum defining $S({r_\Lambda}{L_I})$ converges strongly on $\mathcal{H}_{\omega}$, we start considering $\tilde{f} = G_{\Lambda}h$ where $h$ is a generic element of $C^{\infty}_0(M)$ and we analyze
\[
\| \pi_{\omega }(\sum_n  S_n({r_\Lambda}{L_I}))\pi_\omega(W_+(\tilde{f})) \Psi_\omega \| , \qquad \tilde{f}=G_{\Lambda}h,\,
h\in C^{\infty}_0(M)
\] 
where the norm is induced by the scalar product of $\mathcal{H}_\omega$.
We know from Lemma \ref{le:pullback} that
$r_{\Lambda} W_{\Lambda}(h)= W_+(\tilde{f})$, and thus, recalling the form of the action of $r_{\Lambda}$ on $S$ given in equation \eqref{eq:repS} we obtain 
\begin{align*}
\| \pi_{\omega }(\sum_n  S_n({r_\Lambda}{L_I}))\pi_\omega(W_+(\tilde{f})) \Psi_\omega \|^2
&=
\omega(W_+(\tilde{f})^*(\sum_k  S_k({r_\Lambda}{L_I}))^* \sum_n  S_n({r_\Lambda}{L_I}))W_+(\tilde{f}))
\\
&=
\omega_{\Lambda}(W_\Lambda(h)^*(\sum_k  S_{k,\Lambda}({L_I}))^* \sum_n  S_{n,\Lambda}({L_I}))W_\Lambda(h))
\end{align*}

where in the last equation we used the state $\omega_\Lambda$ which is the pullback under $r_\Lambda$ of the vacuum state on the 
$(\mathcal{W}_0, \star_{\Delta_{+}},*)$, 
see Lemma \ref{le:pullback-state}.
Hence, recalling that on 
$(\mathcal{W}_\Lambda, \star_{\Delta_{+,\Lambda}},*)$,
 $\omega_\Lambda$ corresponds to the evaluation on the vanishing configuration, 
\begin{align*}
\| \pi_{\omega }(\sum_n  S_n({r_\Lambda}{L_I}))\pi_\omega(W_+(\tilde{f})) \Psi_\omega \|^2
&=
\left| \ev{W_\Lambda(h)^*(\sum_k  S_{k,\Lambda}({L_I}))^* \sum_n  S_{n,\Lambda}({L_I}))W_\Lambda(h)}_0\right|
\end{align*}
denoting by $w_s$ the symmetric part of $\Delta_{+,\Lambda}$, computing the product present above,  we have
\begin{align*}
\| \pi_{\omega }(\sum_n  S_n({r_\Lambda}{L_I}))\pi_\omega(W_+(\tilde{f})) \Psi_\omega \|^2
&\leq 
\sum_{n,k} \frac{1}{n!k!} \left| e^{-w_s(h,h)} 
\ev{ e^{-\mathrm{i}\Phi(h)}\cdot_{w_s}  
\prod_k^{\cdot_{w_s}} {\overline{L_I}}
\cdot_{w_s}
\prod_n^{\cdot_{w_s}} {{L_I}}
\cdot_{w_s} e^{\mathrm{i}\Phi(h)}
}_0\right|.
\end{align*}
Using Lemma \ref{le:estimate} we obtain
\begin{align*}
\| \pi_{\omega }(\sum_n  S_n({r_\Lambda}{L_I}))\pi(W_+(\tilde{f})) \Psi_\omega \|^2
&\leq 
\sum_{n,k} \frac{1}{n!k!} 
\|g\|_1^{n+k}
\left(\int_{\mathbb{R}} |f(a)|e^{\frac{a^2}{2} W} \dvol  a\right)^{n+k}
\\
&\leq 
\exp\left({ 2 \|g\|_1\int_{\mathbb{R}} |f(a)|e^{\frac{a^2}{2} W} \dvol a }\right)
\end{align*}
This estimate is uniform in $h$, hence it holds on every element of $\mathcal{W}_{\Lambda}$.
Furthermore, by Lemma \ref{le:pullback}
$\pi_{\omega}(r_\Lambda (\mathcal{W}_\Lambda)) \Psi_\omega$
 is dense in $\mathcal{H}_\omega$, hence, we have that the sum defining $S(r_{\Lambda}(L_I))$ converges in the strong operator topology to a bounded operator.
Finally, we observe that it converges to a unitary operator because $S_\Lambda^{-1}(L_I)$ is equal to $S_\Lambda^*(L_I)$ 
thanks to the fact that $L_I$ is local. Hence, 
$1 = S_\Lambda^{-1}(L_I) \star_\Lambda 
S_\Lambda(L_I) = S_\Lambda^{*}(L_I) \star_\Lambda 
S_\Lambda(L_I)$ and then taking $r_\Lambda$ on both sides
$1 = r_\Lambda(S_\Lambda^{*}(L_I) \star_\Lambda 
S_\Lambda(L_I)) = r_\Lambda(S_\Lambda(L_I))^{*} \star  
r_\Lambda(S_\Lambda(L_I))$. Finally, since $\pi_\omega$ is a $*$-homomorphism, the thesis follows. 
\end{proof}

\subsection{The cutoff removal in the weak-* topology}\label{se:weak-convergence}

%
%
%
%
%
%
%

In this section we shall specialize the analysis to the following interaction Lagrangian $L_I(g,f_\Lambda)$ 
where 
\begin{equation}\label{eq:f-special}
f_\Lambda(a) = \frac{\lambda}{\Lambda} \frac{\partial^4}{\partial a^4}e^{-\frac{a^2}{4 \sigma \Lambda^2}}.
\end{equation}
which corresponds to a Lagrangian of the form $\phi^4$ in the limit of vanishing $\Lambda$. Furthermore, to keep the analysis simple, we use the same parameter $\Lambda$ already used in the regularization of the various propagators constructed above like $\Delta_{+,\Lambda}$. 
Similarly, the spacetime dimension is here restricted to $d\in\{2,3,4\}$.
With this restriction we have that under scaling of $\Lambda$ $W=\Delta_{+,\Lambda}(x,x)< C \Lambda^{-2}$.

We consider two compactly supported smooth functions $h$ and $g$ and we analyze
\[
S_\Lambda(h)
= S(r_\Lambda L_{I}(f_\Lambda,g))^{-1}\star 
S(r_\Lambda L_{I}(f_\Lambda,g)+\Phi(h))
\]
where $L_I(f_\Lambda,g)$ is given in \eqref{eq:LI-vertex} and $r_\Lambda$ is introduced in \eqref{eq:rlambdapullback} and it implements the nonlocal smearing, computing the convolution of the field configurations with the distribution given in \eqref{eq:G-lambda}.

We consider the Minkowski vacuum $\omega$ and its GNS triple $(\mathcal{H},\pi,\Psi)$ and we discuss the removal of the cutoff $\Lambda$, namely the limit where $\Lambda\to0$, in the weak operator topology of $\pi S_\Lambda(h)$ in $\mathcal{H}$. 
According to Theorem \ref{th:convergence-unitary1}, $\pi S(r_\Lambda L_{I}(f_\Lambda,g))$ is a unitary operator for every $\Lambda>0$ and with a similar argument we obtain that also  
$
\pi S(r_\Lambda L_{I}(f_\Lambda,g)+\Phi(h))
$ 
is unitary for every $\Lambda$. Hence $\pi S_\Lambda(h)$ is unitary because it is a product of unitary operators. 
To analyse the limit of vanishing regularizations, we can consider sequences of cutoffs $\Lambda \to \Lambda/n$ for positive integers $n$ and we analyze the corresponding
\[
U_{h,n} = \pi S_{\Lambda/n}(h).
\]
The unit ball in $\mathcal{B}(\mathcal{H})$ is compact in the weak operator topology, hence from $\{U_{h,n}\}_n$, a sequence of unitary operators, we can extract a convergent subsequence indexed by $n_k$. 
The limit point could depend on the particularly chosen subsequence, so uniqueness is not guaranteed by this procedure. At the same time the operator product is not continuous with respect to the weak operator topology hence the operator obtained in the limit could be nonunitary and it is not obvious if causality is satisfied by the limit points of these sequences. 
%
%
%
In this section we can nevertheless check that the operators obtained in the limits are nonvanishing in some cases and satisfy some remnant of causality.
%
%
%
%

To analyze convergence in the weak operator topology we make use of Lemma \ref{le:pullback} 
and for $q_i\in \mathcal{S}(M) $ we use the seminorms
\begin{equation}\label{eq:seminorms}
p_{q_1,q_1}(\pi A) = |\omega(r_\Lambda W_\Lambda (-q_1) A r_\Lambda W_\Lambda(q_2))|. 
\end{equation}
By the density argument proved in Lemma \ref{le:pullback} the topology induced by these seminorms in $\mathcal{H}$ coincides with the weak operator topology.  
We have the following. 

\begin{proposition}
Fix the Minkowski spacetime dimension $d\in\{2,3,4\}$.
Let $\omega$ be the Minkowski vacuum and $(\mathcal{H},\pi,\Psi)$ its GNS triple.
Let $g,h\in C^\infty_0(M)$ be smearing functions and $f_\Lambda$ as in \eqref{eq:f-special} with a sufficiently small $\sigma$ so that 
\[
\int |f_\Lambda(a)| e^{\frac{a^2W}{2}}\dvol a <\infty
\]
for any $\Lambda$ where $W=\Delta_{+,\Lambda}(x,x)$. 
Denote by 
\[
U_{h,\Lambda}:=\pi S(r_\Lambda L_I(f_\Lambda,g))^{-1}\star S(r_\Lambda L_I(f_\Lambda,g) + \Phi(h))
\] 
the relative $S$ matrix of a regularized theory represented as an operator on $\mathcal{H}$.
Denote by 
\[
R_{h}(t)=\pi S(\Phi(h))
\]
the time-ordered exponential of $\Phi(h)$. 
Fix $\Lambda$ and consider a sequence $U_{h,\Lambda/n}$ for $n>0\in \mathbb{N}$. There exists a subsequence indexed by $n_k$ such that 
\[
U_{k}=U_{h,\Lambda/n_k}
\]
converges in the weak operator topology to $U$ with $\|U\|\leq 1$ as $k$ tends to infinity. 
Then it holds that:

(i) (Non vanishing) If $\supp{h} \cap J(\supp{g})=\emptyset$ (even if there are constant Minkowski time surfaces which intersect both supports)
we have that for every seminorm 
\[
p_{q_1,q_2} (U_k-R_{h}) \leq \frac{C}{k}
\]
and taking the limit $k\to\infty$ in the weak operator topology $U=R_h$ and $U$ is nonvanishing. 
%
%
%

(ii) (Remnant of causality) Consider $l\in C^\infty_0(M)$ such that the support of $l$ is spatially related to $\supp{h}$ and to $\supp{g}$, without restriction on the support of $h$ and $g$, then for every seminorm
\[
p_{q_1,q_2}([\pi\Phi(l), U_k])   \leq  \| [\pi \Phi(l), U_{k}] \| \leq \frac{C}{k}
\]
hence taking the weak limit $k\to\infty$ gives
\[
[\pi\Phi(l), U]=0
\]
and $U$ satisfies a remnant of causality.
\end{proposition}

\begin{proof}
By an application of Theorem \ref{th:convergence-unitary1}, for every finite $\Lambda>0$, and every $n$, $U_{h,\Lambda/n}$ is unitary because it is a product of unitary operators $U_{h,\Lambda/n} = \pi_\omega S(r_\Lambda L_I(f,g))^* \pi_\omega S(r_\Lambda L_I(f,g) + \Phi(h))$. 
The unit ball in $\mathcal{B}(\mathcal{H})$ is complete in the weak operator topology, hence we can extract a subsequence $\{U_k=U_{h,\Lambda/{n_k}}\}_k$ from $\{U_{h,\Lambda/n}\}_n$ indexed by $n_k$ which is convergent in the weak operator topology to $U$.
%
%

\bigskip
To prove (i) we observe that 
\[
U_k-R_{h}\to U-R_{h}
\]
in the weak operator topology. For every seminorm $p_{q_1,q_2}$ we have that
\[
p_{q_1,q_1}(U_k-R_{h}) \leq \|U_k-R_{h}\|. 
\]
%
%

Using arguments similar to those used in the proof of Theorem \ref{th:convergence-unitary1},
we find the operators $\mathcal{V}(t)$ and $\mathcal{F}(t)$ on $\mathcal{H}$ such that
\[
\pi S_\Lambda(L_I(f,g)) = \mathcal{T} \exp \mathrm{i}\int \mathcal{V}(t) \dvol t
\]
and
\[
\pi S_\Lambda(L_I(f,g)_\Lambda+\Phi(h)) = \mathcal{T} \exp \mathrm{i} \int \left(\mathcal{V}(t)+\mathcal{F}(t)\right) \dvol t.
\]
For every $t$, $\mathcal{V}(t)$ and $\mathcal{F}(t)$ are self-adjoint operators on $\mathcal{H}$ and as proved in Theorem \ref{th:convergence-unitary1} $\mathcal{V}(t)$ is bounded.
%
%
We write the time-ordered exponentials of a generic time-dependent self-adjoint operator $\mathcal{A}(t)$ on $\mathcal{H}$
\[
\mathcal{U}_\mathcal{A}(t,t_0) = \mathcal{T} \exp \mathrm{i} \int_{t_0}^t \mathcal{A}(u) \dvol{u}.
\]
These are the solutions of the time-ordered differential equations
\[
\frac{d}{dt}\mathcal{U}_\mathcal{A}(t,t_0) = \mathrm{i} \mathcal{A}(t)\mathcal{U}_\mathcal{A}(t,t_0), \qquad \mathcal{U}_\mathcal{A}(t,t)=1
\]
and thus $\mathcal{U}_\mathcal{A}(t,t_0)$ are suitable unitary operators on $\mathcal{H}$. 
Differentiating and integrating the operators $\mathcal{U}_\mathcal{A}(t,s)\mathcal{U}_\mathcal{B}(s,t_0)$ in $s\in (t_0,t)$, we obtain the Duhamel formula 
\begin{equation}\label{eq:duhamel-finite}
\mathcal{U}_\mathcal{A}(t,t_0) - \mathcal{U}_\mathcal{B}(t,t_0) = -\mathrm{i} \int_{t_0}^t 
\mathcal{U}_\mathcal{A}(t,s)(\mathcal{A}(s)-\mathcal{B}(s))\mathcal{U}_\mathcal{B}(s,t_0)\dvol s
\end{equation}
and since $\mathcal{U}_\mathcal{A}$ are unitary, the estimate
\begin{equation}\label{eq:duhamel}
\|\mathcal{U}_\mathcal{A}(t,t_0) - \mathcal{U}_\mathcal{B}(t,t_0)\| \leq  \int_{t_0}^t 
\|\mathcal{A}(s)-\mathcal{B}(s)\|
\dvol s
\end{equation}
follows.
%
%
Consider now 
\[
\mathcal{W}(t,t_0)=\mathcal{U}_{\mathcal{V}}(t,t_0)^*  \mathcal{U}_{\mathcal{V}+\mathcal{F}}(t,t_0) 
\]
and introduce the family of self-adjoint operators 
\[
\mathcal{F}_\mathcal{V}(t)=\mathcal{U}_{\mathcal{V}}(t,t_0)^* \mathcal{F}(t)\mathcal{U}_{\mathcal{V}}(t,t_0)
\]
which are the solutions of
\[
\frac{d}{dt}\mathcal{W}(t,t_0)=\mathrm{i}\mathcal{F}_\mathcal{V}(t)\mathcal{W}(t,t_0) , \qquad    \mathcal{W}(t_0,t_0)=1.
\]
Hence
\[
\mathcal{W}(t,t_0) = \mathcal{T} \exp \mathrm{i} \int_{t_0}^t \mathcal{F}_{\mathcal{V}}(u) \dvol{u}.
\]
Differentiating with respect to $u$ and integrating in $(t,t_0)$ the expression $\mathcal{U}_{\mathcal{V}}(u,t_0)^* \mathcal{F}(t)\mathcal{U}_{\mathcal{V}}(u,t_0)$ we get
\begin{equation}\label{eq:generators1}
\mathcal{F}_\mathcal{V}(t) -\mathcal{F}(t) = \mathrm{i}\int_{t_0}^t \mathcal{U}_{\mathcal{V}}(u,t_0)^* [\mathcal{F}(t),\mathcal{V}(u)]\mathcal{U}_{\mathcal{V}}(u,t_0)  \dvol u
\end{equation}
hence
\begin{equation}\label{eq:stima-generators}
\|\mathcal{F}_\mathcal{V}(t) -\mathcal{F}(t)\| \leq \int_{t_0}^t \|[\mathcal{F}(t),\mathcal{V}(u)]\| \dvol u .
\end{equation}
Observing that 
\begin{equation}\label{eq:time-ordered-U}
U_{h,\Lambda} = \lim_{\substack{t_0\to-\infty \\ t_1\to\infty}}
\mathcal{U}_\mathcal{V}(t_1,t_0)^{-1}\mathcal{U}_{\mathcal{V}+\mathcal{F}}(t_1,t_0)=
\mathcal{U}_{\mathcal{F}_\mathcal{V}}(t_1,t_0), \qquad R_{h} = \mathcal{U}_{\mathcal{F}}(t_1,t_0)
\end{equation}
and combining \eqref{eq:duhamel} with \eqref{eq:stima-generators} we get that
\begin{equation}\label{eq:estimate-commutator-retarded}
\| U_{h,\Lambda}-R_{h}\| \leq \lim_{\substack{t_0\to-\infty \\ t_1\to\infty}} \int_{t_0}^{t_1} \dvol s\int_{t_0}^s \|[\mathcal{F}(s),\mathcal{V}(u)]\| \dvol u .
\end{equation}
The limits are harmless because the support in time of the operators is bounded.  
To get an estimate of the right-hand side we denote by 
\[
{E}_\Lambda = \int_{-\infty}^{\infty} \dvol s\int_{-\infty}^s  \dvol u  \|[\mathcal{F}(s),\mathcal{V}(u)]\|
\]
We have that if the support of $g$ is spatially separated from the support of $h$, ${E}_\Lambda$ tends to $0$ when $\Lambda$ vanishes.
Actually, denoting by 
\[
\Tilde\Delta_\Lambda(x-y) =  G_{\Lambda}*\Delta(x-y)
\]
with $G_{\Lambda}$ given in \eqref{eq:G-lambda}, and with $\Delta$ the commutator function (causal propagator),
we have that
\begin{equation}\label{eq:estimate-rLambda}
|{E}_\Lambda| \leq  \int \dvol a \dvol x\; e^{a^2W/2} |f_\Lambda(a)| |a| |g(x)| |\int \dvol{y}\tilde{\Delta}_\Lambda(x-y) h(y)|. 
\end{equation}
For spatially separated $g$ and $h$ it holds that $\tilde\Delta_\Lambda(x-y)$ does not satisfy strict locality, however it can be bounded by any power of $\Lambda$, namely, for every $N \in \mathbb{N}$ there exists $C_N$ such that
\[
| \tilde{\Delta}_\Lambda(t,\mathbf{x})|  \leq  \Lambda^{N} C_N    \qquad \forall (t,\mathbf{x})\in M, \quad \text{such that} \quad |\mathbf{x}|-|t|>\delta>0.
\]
Furthermore, 
recalling that $|r_\Lambda \int \dvol a e^{\mathrm{i}a\Phi(x)} |$ can be bounded by $|f(a)|e^{a^2W/2}$ uniformly in $x$, 
using the restriction on the spacetime dimension and the scaling properties of $W=\Delta_{+,\Lambda}$ which is bounded by
$\{-C \log\Lambda,C \Lambda^{-1},C \Lambda^{-2}\}$ respectively for $d\in\{2,3,4\}$, we obtain that
\[
\int  |f_\Lambda(a)| e^{a^2 W/2} \dvol a \leq \frac{C}{\Lambda^4}.
\]
which diverges for $\Lambda\to0$ but it can be controlled together with $\tilde\Delta_\Lambda$. 
Hence, since by hypothesis the supports of $h$ and $g$ are not spatially separated, we have that there exists $C$ such that 
%
\[
|E_{\Lambda}| \leq C \Lambda. 
\]
Hence, combining these estimates, for every seminorm $p_{q_1,q_2}$
\[
p_{q_1,q_2}(U_k-R_{h})\leq \frac{C}{k}.
\]
Recalling the discussion above, we have that $p_{q_1,q_2}(U-R_h)$ is controlled by
$p_{q_1,q_2}(U_k-R_{h})$ and the latter vanishes in the limit $k\to\infty$. This concludes the proof of (i).

\bigskip
To prove (ii), we denote $\Phi=\pi_\omega\Phi(l)$, and we analyse  $[\Phi, U_k]$ which converges to $[\Phi, U] $ 
in the weak operator topology. For every seminorm $p_{q_1,q_2}$ we have that
\[
p_{q_1,q_1}([\Phi,U_k]) \leq \|[\Phi,U_k]\|. 
\]
To analyze the latter, as before,  we represent $U_k$ by a time-ordered exponential of time-dependent hamiltonians as in \eqref{eq:time-ordered-U}.
%
%
From \eqref{eq:duhamel-finite}, we get
\begin{equation}\label{eq:Duhamel-commutator}
[\Phi,\mathcal{U}_\mathcal{A}(t,t_0)] = \mathrm{i} \int_{t_0}^t  \mathcal{U}_\mathcal{A}(t,s) [\Phi,\mathcal{A}(s)] \mathcal{U}_\mathcal{A}(s,t_0) \dvol s,
\qquad 
\|[\Phi,\mathcal{U}_\mathcal{A}(t,t_0)]\| \leq \int_{t_0}^t \|[\Phi,\mathcal{A}(s)]\| \dvol s.
\end{equation}
Hence, 
\begin{align*}
[\Phi , U_{h,\Lambda} ] 
=&\lim_{\substack{t\to\infty \\ t_0\to-\infty}}   \mathrm{i} \int_{t_0}^t  \mathcal{U}_{\mathcal{F}_\mathcal{V}}(t,s) [\Phi,\mathcal{F}_\mathcal{V}(s)-\mathcal{F}(s)] \mathcal{U}_\mathcal{\mathcal{F}_\mathcal{V}}(s,t_0) \dvol s
\\&+
\lim_{\substack{t\to\infty \\ t_0\to-\infty}}  \mathrm{i} \int_{t_0}^t  \mathcal{U}_{\mathcal{F}_\mathcal{V}}(t,s) [\Phi,\mathcal{F}(s)] \mathcal{U}_\mathcal{\mathcal{F}_\mathcal{V}}(s,t_0) \dvol s
\end{align*}
and we get the estimate
\begin{equation}\label{eq:estimate-rough}
\|[\Phi , U_{h,\Lambda} ]\| 
\leq  \lim_{\substack{t\to\infty \\ t_0\to-\infty}}  \int_{t_0}^t \|  [\Phi,\mathcal{F}_\mathcal{V}(s)-\mathcal{F}(s)]\| \dvol s
+
 \lim_{\substack{t\to\infty \\ t_0\to-\infty}} \int_{t_0}^t \| [\Phi,\mathcal{F}(s)] \| \dvol s
\end{equation}

%

From \eqref{eq:generators1}
\[
 [\Phi , {\mathcal{F}_\mathcal{V}}(t) -{\mathcal{F}}(t)]= 
 \mathrm{i}\int_{t_0}^t [ \Phi,\mathcal{U}_{\mathcal{V}}(u,t_0)^* [\mathcal{F}(t),\mathcal{V}(u)]\mathcal{U}_{\mathcal{V}}(u,t_0) ] \dvol u,
\]
combining with \eqref{eq:Duhamel-commutator}, we have that 
\[
 \| [\Phi , {\mathcal{F}_\mathcal{V}}(t) -{\mathcal{F}}(t)]\| 
 \leq 
 \int_{t_0}^t\dvol u  \|[\mathcal{F}(t),\mathcal{V}(u)] \| \int_{t_0}^t \dvol s\|[\Phi,\mathcal{V}(s)]\| 
 +
  \int_{t_0}^t \|[ \Phi,[\mathcal{F}(t),\mathcal{V}(u)]]\|\dvol u.
\]
Using the obtained estimate in \eqref{eq:estimate-rough} and observing that $\int \dvol s \mathcal{F}(s) = \pi \Phi(h)$ we conclude that 
\begin{equation}\label{eq:double-commutator}
\begin{aligned}
\|[\Phi , U_{h,\Lambda} ]\| 
\leq &\|[\Phi,\pi \Phi(h)]\| \\
&+
\lim_{\substack{T\to\infty \\ t_0\to-\infty}}
\int_{t_0}^T \dvol t \left(  \int_{t_0}^t\dvol u  \|[\mathcal{F}(t),\mathcal{V}(u)] \| \int_{t_0}^t \dvol s\|[\Phi,\mathcal{V}(s)]\| 
 +
  \int_{t_0}^t \|[ \Phi,[\mathcal{F}(t),\mathcal{V}(u)]]\|\dvol u \right).
\end{aligned}
\end{equation}

%
%
%
%
%
%
%
%
Thus, inequality \eqref{eq:double-commutator}  furnishes an estimate of $p_{q_1,q_2}([\pi\Phi(l), U_k])   \leq  \| [\pi \Phi(l), U_{k}] \|$ 
once we get proper bounds of the various commutators present on the right-hand side of \eqref{eq:double-commutator}. 
We first observe that $[\Phi , \pi\Phi(h) ] = [\pi\Phi(l) , \pi\Phi(h) ] =0$ by causality, furthermore, because of the compactness in time of the supports of $l$, of $g $ and of $h$ the limits for $T$ and $t_0$ to infinity can be taken without any problem.
We observe that estimates similar to those written above for $E_\Lambda$ can be obtained for
\begin{align*}
\|[\Phi,\mathcal{V}(s)]\|, \qquad  
\|[ \Phi,[\mathcal{F}(t),\mathcal{V}(u)]]\|
\end{align*}
Finally $\|[\mathcal{F}(t),\mathcal{V}(u)]\|$ grows at most as $\Lambda^{-3}$ for small $\Lambda$ and this growth can again be controlled by the decay of $\tilde\Delta_\Lambda$ present in  $\|[\Phi,\mathcal{V}(s)]\|$. 
Combining these observations, we get the desired estimates
\[
\|[\Phi,U_k]\|\leq C/k
\]
and the latter allows us to conclude the proof of this item.
\end{proof}
%
%
%
%
\color{black}

\section{Limit {$\Lambda \to 0$}, renormalizability, counterterms and perturbation theory}
\label{se:limits-no-cutoffs}

In this section we discuss how to recover known {results in perturbation theory when suitable limits of vanishing regulators are considered. }
%
%
We discuss two cases: the two-dimensional theories where the UV problems can be treated as in the case of Sine-Gordon theories 
\cite{Froe76,FS}
(see also
\cite{BFM1, Bahns:2016zqj, BahnsPinamontiRejzner})
{ where convergence of the perturbative series can be studied}
 and 
we furthermore discuss how to recover the $\lambda\phi^4_3$ case on a three-dimensional Minkowski spacetime; finally, {we present the problems which arise in the case of a $\lambda\phi^4_4$ theory and we discuss possible solutions}. 

\subsection{2D generic interactions of compact (and small) support in the $a$ domain}\label{se:two-d-limits}

In the two-dimensional case, the theory does not need to be renormalized. 
Furthermore interaction Lagrangians of Sine-Gordon type produce $S$ matrices which converge uniformly in the field configuration also in the limit where $\Lambda$ vanishes in $\Delta_{+,\Lambda}$ if the support of $f(a)$ in the interaction Lagrangian is sufficiently small. 

To this end, we consider
\[
L_{I} = \int  e^{\mathrm{i} a \phi(x) } f(a) \dvol a \; g(x) \dvol x
\]
where $f(a)$ is smooth and of compact support, furthermore, 
\[
\supp f\subset [-\tilde{a},\tilde{a}]
\]
where $\tilde{a}^{2}  < 4\pi $
and 
\[
\supp g \subset D_{\mu}
\]
where $D_\mu=\{(t,x)\in M_2\, |\, -\mu<t+x<\mu, -\mu<t-x<\mu\}$ and $\mu$ is a fixed parameter. 
These bounds on $\tilde{a}$ and on $\mu$ allows us to bound the $S$ matrix uniformly in $\Lambda$, in an equivalent way to the treatment of the ultraviolet finite regime in the Sine-Gordon theory \cite{Froe76, FS}.

The difference with the Sine-Gordon theory is in the fact that here products of vertex operators with different $a$ appear in the $S-$matrix. 
To estimate these products, we observe that the antisymmetric part of the products produces inessential phases, we shall thus use Hölder inequality which holds in the case of  commutative spacetimes to bound them, see \cite{BahnsPinamontiRejzner} for further details.
Afterwards each factor in the product can 
be bounded uniformly in $\Lambda_2$ by means of a Cauchy determinant formula, see e.g. Lemma 2.7 in \cite{BahnsPinamontiRejzner} and references therein.

In this way we get an estimate that works uniformly in $\Lambda$ (the regularizing parameter in $\Delta_{F,\Lambda}$).

We shall be more precise in the proof of the following Theorem. 
\begin{theorem}
Consider the two-dimensional Minkowski space $M_2$. Let $\Delta_+$ be the two-point function of the Minkowski vacuum of the massive Klein Gordon field of mass $m>0$. 
Consider the interaction Lagrangian
\[
L_I(g,f)(\phi)= \int_{M_2} \dvol x\, g(x)\int_{\mathbb{R}} \dvol a f(a) e^{\mathrm{i} a \phi(x)}
\]
where $f\in C^{\infty}_0((-\tilde{a},\tilde{a}))$ with ${\tilde{a}}^2< 4\pi$, $\overline{f}(a)=f(-a)$
and with an adiabatic cutoff $g\in C^\infty_0(D_\mu)$ where $D_\mu=\{(t,x)\in M_2\, |\, -\mu<t+x<\mu, -\mu<t-x<\mu\}$. The sum defining the $S$ matrix 
\[
S(L_I)= \sum_{n\geq 0} \lambda^n S_n
= \sum_{n\geq 0} \frac{\mathrm{i}^n\lambda^n}{n!} \mathcal{T}(\underbrace{L_I\otimes \dots \otimes L_I)}_{n}
\]
is weakly convergent.  
Actually, there exist two positive constants $C$ and $K$ such that for $p\in [1,\frac{4\pi}{\tilde{a}^2})$ and $\frac1p+\frac1q=1$,
the following estimate holds
\[ |S_n| \leq 
\frac{ \|f\|_1^n}{n!} e^{n\frac{K}{2} \tilde{a}^{2}}
 (2\mu)^{n \frac{\tilde{a}^2}{4\pi}} \|g\|_q^n (C^n n!)^{\frac{1}{p}}
\]
uniformly for $\phi \in \mathcal C$.
Hence, the series which defines $S(L_I)$ when tested on a field configuration $\phi$ is absolutely convergent uniformly in $\phi\in\mathcal{C}$.
\end{theorem}
\begin{proof}
We observe that  
\begin{align*}
    S_n &= 
\frac{\mathrm{i}^n}{n!} \mathcal{T}(\underbrace{L_I\otimes \dots \otimes L_I)}_{n}
=\frac{\mathrm{i}^n}{n!} (\underbrace{L_I\cdot_T \dots \cdot_T L_I)}_{n}
\\
&=
\frac{\mathrm{i}^n}{n!}  \int_{[-\tilde{a},\tilde{a}]^n} \dvol a\, 
f(a_1)\dots f(a_n)
V_{a_1}\cdot_{T}\dots \cdot_{T} V_{a_n}
\end{align*}
where $V_{a_i} = V_{a_i}(g)$.
Denoting the symmetric part of the two-point function by  
$w_s = \Delta_+-\frac{i}
{2}\Delta$, we observe that its integral kernel is  
\begin{align*}
w_s(t,{x}) &= \frac{1}{2\pi} \text{Re}(K_0(m\sqrt{-t^2+x^2}))
\\
&= -\frac{1}{4\pi} 
\log(\frac{|-t^2+x^2|}{4\mu^2}) + r(t,{x})\\
&= w_{s0}(t,{x}) + r(t,{x})
\end{align*}
where $r$ is a continuous remainder.  
With this mild singularity, and with the choices of $g$ and $f$ in $L_I$, the time-ordered products are well defined. To prove it, we observe that 
\[
|V_{a_1}\cdot_{T}\dots \cdot_{T} V_{a_n}|
\leq 
|\langle V_{a_1}\cdot_{w_s}\dots \cdot_{w_s} V_{a_n}\rangle_0|,
\]
then, {by the Cauchy-Schwarz inequality ($|( 1,A)| \leq \|A\|\|1\|=|( A,A)|^{\frac{1}{2}}=|\langle A^*\cdot_{w_s} A\rangle_0|^{\frac{1}{2}}$)
\[
|\langle V_{a_1}\cdot_{w_s}\dots \cdot_{w_s} V_{a_n}\rangle_0|
\leq 
|\langle (V_{a_1}\cdot_{w_s}V_{-a_1})\cdot_{w_s} \dots \cdot_{w_s} (V_{a_n} \cdot_{w_s} V_{-a_n})\rangle_0|^{\frac{1}{2}}.
\]
H\"older inequality holds with respect to the product $\cdot_{w_s}$ because it is commutative, $\ev{\cdot}_0$ is a positive state and $(V_{a_j}\cdot_{w_s}V_{-a_j})$ are formally self-adjoint.
It implies that 
\[
|\langle (V_{a_1}\cdot_{w_s}V_{-a_1})\cdot_{w_s}\dots \cdot_{w_s} (V_{a_n} \cdot_{w_s} V_{-a_n})\rangle_0|^{n}
\leq 
\prod_{j=1}^n |\langle \underbrace{(V_{a_j}\cdot_{w_s}V_{-a_j})\cdot_{w_s}\dots \cdot_{w_s} (V_{a_j} \cdot_{w_s} V_{-a_j})}_n\rangle_0|.
\]
We then apply conditioning and inverse conditioning, see e.g. Theorem 2.2 in \cite{BahnsPinamontiRejzner}, to bound 
\[
|\langle (V_{a_j}\cdot_{w_s}V_{a_{-j}})^{\cdot_{w_s} n}\rangle_0|^{\frac{1}{2}}
\leq
e^{n a_j^2 \frac{K}{2}}
|\langle (V_{a_j}\cdot_{w_{s0}} V_{a_{-j}})^{\cdot_{w_{s0}}n}\rangle_0|^{\frac{1}{2}}
\]
where the positive finite constant $K$ is chosen in such a way that it bounds the coinciding point limit of the negative part of $w_s-w_{s0}$. It satisfies the bound   
\[
\lim_{x\to 0} \left|w_s(x)-w_{s0}(x)\right| 
= \frac{1}{4\pi} \left|
2\gamma+\log\frac{m^2\mu^2}{2}
\right| < K <\infty
\]
where $\gamma$ denotes the Euler-Mascheroni constant.
We can then use the Cauchy determinant lemma on every factor of the product indexed by $j$ to find that 
\[
|\langle (V_{a_j}\cdot_{w_{s0}} V_{a_{-j}})^{\cdot_{w_{s0}}n}\rangle_0|^{\frac{1}{2}}
\leq  (2\mu)^{n \frac{\tilde{a}^2}{4\pi}} \|g\|_q^n (C^n n!)^{\frac{1}{p}}
\]
%
for a suitable constant $C$, (see e.g. Lemma 2.7 in \cite{BahnsPinamontiRejzner}) and where $\tilde{a}$ bounds $|a_j|$ on the support of $f$.
Combining these estimates we get 
\[
|V_{a_1}\cdot_{T}\dots \cdot_{T} V_{a_n}|\leq 
|\langle V_{a_1}\cdot_{w_s}\dots \cdot_{w_s} V_{a_n}\rangle_0|
\leq 
e^{ n\frac{K}{2} \tilde{a}^{2}}
 (2\mu)^{n \frac{\tilde{a}^2}{4\pi}} \|g\|_q^n (C^n n!)^{\frac{1}{p}}.
\]
Using this estimate in $S_n$ we get the statement of the theorem.
}
\end{proof}

\subsection{The three-dimensional case for an  interaction Lagrangian
 $\lambda \phi^4$}\label{se:limit3d}

In this section we discuss how to recover the {perturbative construction of } $\lambda \phi^4_3$ theory in suitable limits. 
To this end we add to the interaction Lagrangian a suitable parameter that later will be removed.

The interaction Lagrangian we are working with is assumed to have the form
\[
L_{I,\Lambda_1,\Lambda_2} := \lambda \int_M   e^{-\Lambda_1^2 \phi^2}\left(\frac{\phi^4}{4} + a_1(\Lambda_2) \frac{\phi^2}{2} + a_2(\Lambda_2)\right)
g(x)\dvol x
\]
where the coefficients $a_i$ depend on the cutoffs and will be fixed later. 
Furthermore $L_I$ is assumed to be an element of
$\mathcal{V}^{\Delta_{+,\Lambda_2}}$.

\subsubsection{Renormalization of $\lambda \phi^4_3$}
In $d=3$ the theory with the interaction Lagrangian density of the form $\lambda \phi^4$ is superrenormalizable, actually in the expansion of the $S$ matrix of the theory as a sum over all possible Feynman diagrams, there is only one connected graph with two external legs which is nonvanishing and two connected graphs with no external legs which need to be renormalized. 
The corresponding distributions that need to be renormalized in the limit of vanishing $\Lambda$ are
\begin{align*}
G_1 &= \int \dvol x\dvol y\;\phi(x)\phi(y)\Delta_{F,\Lambda}(x,y)^3g(x)g(y),
\\
G_2 &=
\int \dvol x\dvol y \;\Delta_{F,\Lambda}(x,y)^4g(x)g(y)
, 
\\
G_3  &=
\int \dvol x\dvol y\dvol z \; \Delta_{F,\Lambda}(x,y)^2\Delta_{F,\Lambda}(x,z)^2\Delta_{F,\Lambda}(y,z)^2g(x) g(y) g(z).
\end{align*}
The degree of divergence of $G_1$ and $G_3$ is $0$, while the degree of divergence of $G_2$ is 1.
The graph with two external legs is the sunrise diagram.
We can add a mass counterterm to $L_I$ to make $S(L_I)$ regular in the limit of vanishing $\Lambda$.
The two graphs with vanishing external legs are of order 2 and 3 in powers of $L_I$.
The renormalization in this case amounts to adding a constant term in the fields to the Lagrangian. 
This term has no effect on the Bogoliubov map and amounts to a phase in $S(L_I)$. 
Notice that in this discussion we have discarded graphs containing tadpoles because they do not appear in the construction we have presented.

Hence, we consider an interaction Lagrangian 
\begin{equation}\label{eq:LIPphi43}
L_{I,\Lambda_1,\Lambda_2}(g) :=  \int_{M}
e^{-\Lambda_1^2 \phi^2(x)}  \left( \lambda \frac{ \phi^4(x)}{4} + \delta m(\Lambda_2)(x) \frac{\phi^2(x)}{2}  +c(\Lambda_2)(x) \right) g(x) \dvol x
\end{equation}
according to {the renormalization scheme for Lorentzian spacetime studied e.g. in \cite{BrunettiFredenhagen00, HW02} or to the dimensional regularization for Lorentzian theories analyzed e.g. in \cite{Dutsch:2014aa} which is also compatible with}
the Bogoliubov-Parasiuk-Hepp-Zimmermann (BPHZ) renormalization scheme \cite{BogoPara,Hepp,Zimmermann}, see also \cite{Collins} and \cite{BogoliubovShirkov}, we have that 
{
%
%
%
%
%
\begin{align*}
 \delta m(\Lambda_2)(x) = & -6 \mathrm{i}\lambda^2 \int \left(\Delta_{F,\Lambda_2}^3(x,y)-{\Delta_{F,\Lambda_2}^3}_{\text{reg}}(x,y) \right)g(y)\dvol y
\\
c(\Lambda_2)(x) = & -\frac{3}{4}\mathrm{i}\lambda^2 \int \left(\Delta_{F,\Lambda_2}^4(x,y) - {\Delta_{F,\Lambda_2}^4}_{\text{reg}}(x,y)\right)g(y) \dvol y+
\\
& +\frac{9}{2}
\lambda^3\int 
\left(\Delta_{F,\Lambda_2}^2(x,y)\Delta_{F,\Lambda_2}^2(x,z)\Delta_{F,\Lambda_2}^2(y,z)
-
{\Delta_{F,\Lambda_2}^2\Delta_{F,\Lambda_2}^2\Delta_{F,\Lambda_2}^2}_{\text{reg}}(x,y,z)
\right)\cdot\\
&\cdot g(y)g(z)\dvol y \dvol z .
\end{align*}
Here,  the distributions present in the previous expressions are formed by two terms. The first is singular in the limit of vanishing $\Lambda_2$ and the second  is one of its regularizations, namely the latter is a regular distribution also in the limit of vanishing $\Lambda_2$. These regularized distributions can be constructed as e.g. in \cite{HW02} or as in \cite{BrunettiFredenhagen00}. These regularized distributions differ from the unregularized ones by contributions supported on the thin diagonal, hence local.
Furthermore, the regularization can be obtained keeping the properties stated by Hollands and Wald in \cite{HollandsWald2001,HW02} hence keeping unitarity of the associated $S$ matrix. We refer to \cite{HW02} for further details.  
%
%
}
%
%
With this choice, the limit $\Lambda_1$ to $0$ can be taken and, as we shall see below, one recovers ordinary perturbation theory.
In the limit of vanishing $\Lambda_1$  we obtain an interaction Lagrangian which is already known in the literature, see e.g. equation (2.1.1) in \cite{MagnenSenor}.

We recall here that there is an extra renormalization freedom in the choice of the parameters $\delta{m}$ and $c$, which here we have implicitly fixed. Other choices of the renormalization constants would result in a change of both $\delta m$ and $c$ by some finite constant which however does not alter the construction.

Finally we observe that thanks to the fact that $\delta m$ is quadratic in the coupling parameter $\lambda$ and $c$ is the sum of a quadratic and cubic contribution in $\lambda$ we have that $L_I$ can be expanded in the following way with respect to the coupling parameter
\begin{equation}\label{eq:expansionV}
L_{I,\Lambda_1,\Lambda_2} = \lambda L_{1;\Lambda_1,\Lambda_2} + \lambda^2 L_{2;\Lambda_1,\Lambda_2} + \lambda^3 L_{3;\Lambda_1,\Lambda_2} .
\end{equation}
With this in mind we have that the time-ordered exponential of products of $L_{I,0,0}$ is order by order in perturbation theory finite.
We furthermore have the following corollary of Theorem \ref{th:convergence-Smatrix} whose proof is an adaptation to the three-dimensional case and can thus be omitted 
\begin{corollary}\label{cor:convergenceP34}
    For every $\Lambda_1>0$ and $\Lambda_2>0$, 
    {such that
    \begin{equation}\label{eq:condition-Lambda}
    2\Lambda_1^2 \Delta_{+,\Lambda_2}(x,x) < 1,
    \end{equation}
    }
    the perturbative series defining $S(L_{I,\Lambda_1,\Lambda_2})$ is absolutely convergent, uniformly on $\phi$.
\end{corollary}

Notice that the condition \eqref{eq:condition-Lambda} is necessary to make sure that the hypotheses of Theorem  \ref{th:convergence-Smatrix} are satisfied and in particular condition \eqref{eq:condition-A} holds.

%
%
%
%
%
%
%
%

\color{black}

\subsubsection{Restoring perturbation theory}

Consider the expansion of the $S$ matrix as a power series in the coupling constant.
\[
S(\Lambda_1,\Lambda_2) = \sum_{n\geq 0} S_n(\Lambda_1,\Lambda_2) \lambda^n
\]
for every $n$, $S_n$ is a suitable element of $\mathcal{F}$, which depends on the parameters $\Lambda_1, \Lambda_2$. 
We also observe that, in view of the finite expansion of $L_{\Lambda_1,\Lambda_2}$ in powers of $\lambda$ given in \eqref{eq:expansionV}, we have that  $S_0=1$ and for $n\geq 0$ 
\begin{equation} \label{eq:Sordern}
S_{n}(\Lambda_1,\Lambda_2) = 
\sum_{p\geq 0} \sum_{J\in \{1,2,3\}^p, \|J\|=n}
\frac{1}{p!} 
\mathcal{T}_{\Lambda_2}(\underbrace{L_{J_1;\Lambda_1,\Lambda_2} \otimes \dots \otimes L_{J_p;\Lambda_1,\Lambda_2} }_{p}),  
\end{equation}
where $\mathcal{T}_{\Lambda_2}$  is the map which realizes the time ordering and it coincides with the product $\cdot_{T,\Lambda_2}$ which is equal to $\star_{\Delta_{F,\Lambda}}$ in \eqref{eq:prodcut}. 
We consider now the limit where $\Lambda_i$ tends to $0$.
We observe that, for every $S_n$ the limit can be taken in suitable orders and that one obtains the ordinary perturbation theory with a precise choice of the renormalization constants.

The coefficients of the power series which we obtain with {one of the renormalization schemes described e.g. in \cite{HW02}
are given by
\[
S_{ren}(L_{I,0,0}) = \sum_{n}
S_{ren}(L_{I,0,0})_n
\lambda^n
\]
with
\begin{equation} \label{eq:S(BPHZ)ordern}
S_{ren}(L_{I,0,0})_n = \lim_{\Lambda_2 \to 0}
\sum_p \sum_{J\in \{1,2,3\}^p, \|J\|=n}
\frac{1}{p!} 
\mathcal{T}_{\Lambda_2}(\underbrace{L_{J_1;0,\Lambda_2} \otimes \dots \otimes L_{J_p;0,\Lambda_2} }_{p}).  
\end{equation}
where $\|J\|=\sum_i J_i$.
}
In this way we obtain the description of $S_n$ in terms of suitable limits of regularized Feynman diagrams.
We also observe that the limit of vanishing  $\Lambda_2$ can be taken without encountering divergences because the terms $L_2$ and $L_3$ in $L_I$ are precisely the counterterms which cancel the singularities at the coinciding point limits. 

\begin{proposition}
{ For some choice of the renormalization parameters, the coefficients of the power series expansion of $S$ obtained with the ordinary renormalization methods are given in equation \eqref{eq:S(BPHZ)ordern}.}
\end{proposition}
The statement of this proposition is known in the literature, see e.g. \cite{Collins}. 

We prove here that at each perturbative order
$S(L_I)_n$
converges to {$S_{ren}(L_{I,0,0})_n$ 
 when the limits 
of vanishing  $\Lambda_i$ are taken.
}

%
%
%
%
%
%
%
%
 
\begin{proposition}\label{th:3d-limits-noLambda}
At each perturbative order, namely for every $n$
\[
\lim_{\Lambda_2\to 0}
\lim_{\Lambda_1\to 0} 
S_n(\Lambda_1,\Lambda_2) = 
{ S_{ren}(L_{I,0,0})_n}
\]
\end{proposition}
\begin{proof}
Let us start observing that 
\begin{align*}
L_{I,\Lambda_1,\Lambda_2} &= 
\int e^{-\Lambda_1^2 \phi^2}  \left( \frac{\phi^4}{4} + \delta m(\Lambda_2) \frac{\phi^2}{2}  +c(\Lambda_2) \right) g(x)\dvol x 
\\
&= 
\int 
r\left(\frac{a}{\Lambda_1}\right)
\frac{1}{\Lambda_1}
\left( \frac{1}{4}\frac{\partial^4}{\partial a^4} - \frac{1}{2}\delta m(\Lambda_2)(x) \frac{\partial^2}{\partial a^2}  +c(\Lambda_2)(x) \right)
e^{\mathrm{i} \phi(x)a}
g(x)\dvol x 
\dvol a 
\\&= 
\int 
r\left(\frac{a}{\Lambda_1}\right)
\frac{1}{\Lambda_1}
D_a(x)
e^{\mathrm{i} \phi(x)a}
g(x)\dvol x 
\dvol a
\end{align*}
where $r$ is the Fourier transform of the factor $e^{-\Lambda_1^2 \phi^2}$, namely  
\[
r(a) = \frac{1}{2\pi} \int e^{- p^2 } e^{-\mathrm{i} p  a}\dvol p
\]
and $D_a(x)$ is the differential operator in $a$ with coefficients depending on $x$ implicitly defined in the last equation. Hence
\begin{gather*}
\mathcal{T}_{\Lambda_2}(\underbrace{L_{I,\Lambda_1,\Lambda_2} \!\otimes\! \dots \!\otimes L_{I,\Lambda_1,\Lambda_2} }_{n})
=
\\
=\! \int \!\!\dvol  a^n \!\!\! 
\int \!\!\dvol x^n \!\left(\!\prod_{i} \!\frac{1}{\Lambda_1}r\left(\frac{a_i}{\Lambda_1}\right) g(x_i)  D_{a_i}(x_i) \!\right) e^{\sum\limits_{j}\mathrm{i} a_j \phi(x_j)} 
\!
e^{\sum\limits_{1\leq i<j\leq n} \!\!\!\! a_ia_j \Delta_{F,\Lambda_2}(x_i,x_j)}. 
\end{gather*}
After having computed all the derivatives present in $D_{a_i}$ for all $a_i$, keeping $\Lambda_2 >0$, we obtain an expression in which, for every $\phi\in \mathcal{C}$ we can take the limit $\Lambda_1\to 0 $.
In the limit of vanishing $\Lambda_1$ the integrals over $r \dvol a$ converge to the Dirac delta function, hence, the result is 
\begin{align*}
\lim_{\Lambda_1\to 0} \mathcal{T}_{\Lambda_2}(\underbrace{L_{\Lambda_1,\Lambda_2} \otimes \dots \otimes L_{\Lambda_1,\Lambda_2} }_{n})
=
\mathcal{T}_{\Lambda_2}(\underbrace{L_{I,0,\Lambda_2} \otimes \dots \otimes L_{I,0,\Lambda_2} }_{n}). 
\end{align*}
Using this relation in the comparison of equations \eqref{eq:Sordern} with 
 equation \eqref{eq:S(BPHZ)ordern}
after taking the appropriate limits in the appropriate order we obtain the thesis.
\end{proof}

\begin{corollary}\label{cor:perturbation-equal}
Under the hypotheses of {Proposition \ref{th:3d-limits-noLambda},} 
similar results hold upon taking the limits of vanishing $\Lambda_i$ in another direction, namely the following holds 
\[
\lim_{\Lambda\to 0}
S_n(\Lambda,\frac{1}{-\log(\Lambda)}) = 
{ S_{ren}(L_{I,0,0})_n}
\]
\end{corollary}
\begin{proof}
Notice that
\[
S_n(\Lambda,\frac{1}{-\log(\Lambda)}) 
\]
can be expanded in powers of $a$. 
Arguing as in the proof of Theorem \ref{th:3d-limits-noLambda} we get that the order $0$ converges to $S_{ren}(L_{I,0,0})_n$.
The terms of order greater than $0$ are of the form $\sum_{n,k}a^{n}\Delta_{F,\Lambda}^{k} C_{nk}$.
However $\Delta_{F,\Lambda }^{k}$ is a continuous function and it is bounded by $C |\log(\Lambda)|^{k}$ where $C$ is a suitable constant. At the same time $a$ gets multiplied by a factor $\Lambda$ .

The $n$th order contribution is thus bounded by $C_{nk} \Lambda^n |\log(\Lambda)|^{k}$ and it tends to $0$ in the limit of vanishing $\Lambda$ for every $k$ in the case $n\geq 1$.
\end{proof}

\subsubsection{Extracting subsequences which converge in the weak-* topology}

{ Similar to the analysis presented in section \ref{se:weak-convergence},}
the following result now allows us to obtain convergent sequences of unitary operators in the GNS representation of the vacuum. 
These unitary operators are various regularizations of the sought $S$-matrix of the theory. 

We stress that we do not have infrared problems in this analysis because the interaction Lagrangian is smeared with a compactly supported smooth function $g$.

\begin{proposition}\label{eq:th-*weak-convergence}
Fix two strictly positive constants, $\Lambda_1$ and $\Lambda_2$ {which satisfy \eqref{eq:condition-Lambda}}. 
Consider an interaction Lagrangian $L_{I,\Lambda_1,\Lambda_2}$ of the form given in \eqref{eq:LIPphi43}, in the three-dimensional Minkowski spacetime.
Let $(\mathcal{H},\pi,\psi)$ be the GNS triple of the vacuum state of the free theory.

Let $\Lambda_1,\Lambda_2$ be two strictly positive constants, and consider $L_{I,\Lambda_1,\Lambda_2}$ as in \eqref{eq:LIPphi43}. The corresponding  $S$-matrix, denoted by 
$S(L_I)(\Lambda_1,\Lambda_2)$,
is described by a suitable unitary operator in $\mathcal{B}(\mathcal{H})$.

Consider a sequence of pairs of constants $I=\{(\Lambda_1/n,\Lambda_2/(\log(n))\}_n$. It is then possible to extract a subsequence of $I$, labeled by $\{n_k\}_k$ which makes the sequence of unitary operators
\[
\{
S(L_I)(\frac{\Lambda_1}{n_k},\frac{\Lambda_2}{\log(n_k)})
\}_{k\in \mathbb{N}}
\]
weakly-$*$ convergent in $\mathcal{B}(\mathcal{H})$ {to an operator $S$ of norm smaller than or equal to $1$.}
\end{proposition}

%

\begin{proof}
We have that 
$S(L_I)(\Lambda_1,\Lambda_2)$ are described by unitary operators in the vacuum Fock space by Theorem \ref{th:convergence-Smatrix}, Corollary \ref{cor:convergenceP34} and Theorem \ref{th:convergence-unitary1}.
The unit ball {is weak-$*$ compact and thus sequentially compact.}
%
%
Hence, considering a suitable subsequence 
$\{ S(L_I)({\frac{\Lambda_1}{n_k},\frac{\Lambda_2}{\log(n_k)}}) \}_{n_k}$ of $I$ we obtain sequences of unitary operators acting on the vacuum Fock space which converge {in the set of bounded operators}. 
\end{proof}

{\bf Remark}
Proposition \ref{eq:th-*weak-convergence}
states that there exists at least one subsequence which converges. It is however possible that various subsequences with different limit points exist. 
We have thus that the limit elements of various subsequences are not necessarily equivalent, in the sense that different choices of subsequences could have different limit points. 
{ At the same time, the operator obtained in the limit could not be unitary because the product is not jointly continuous in the weak-$*$ topology. 
In spite of this fact, the very same limits applied to the coefficients of the power expansion of $S$ furnishe convergent finite coefficients.}
%
\bigskip

{ 
Actually, we observe that the limits of vanishing regulators applied to the coefficients of the expansion in powers of $\lambda$ of $S$ do not depend on the chosen sequence.
%
In other words, the formal power series in $\lambda$ corresponding to the various sequences coincide by an application of Corollary \ref{cor:perturbation-equal}.}
%
We observe that the results of 
{ Proposition  \ref{eq:th-*weak-convergence}}
are not in contrast with Haag's theorem because the interaction Lagrangian which is considered here is always of compact support because of the presence of the smearing function $g$.

\subsection{Towards the four-dimensional case for an  interaction Lagrangian
 $\lambda \phi^4$}\label{se:limit4d}
 
In order to analyze $\lambda \phi_4^4$ theory in the framework presented above we need to solve two problems which become manifest analyzing the form of the renormalized Lagrangian density with counterterms.

Using e.g. {one of the renormalization schemes discussed in \cite{HW02} or \cite{BrunettiFredenhagen00}},
the renormalized Lagrangian density with counterterms of a $\lambda \phi^4_4$ theory takes the form
\[
\tilde{\mathcal{L}}  =  
-\frac{1}{2}\partial_\mu \phi \partial^\mu \phi
 - m^2 \frac{\phi^2}{2} - \lambda\frac{\phi^4}{4} 
- \delta Z \frac{1}{2} \partial_\mu \phi\partial^\mu\phi  - \delta M \frac{\phi^2}{2} -\delta{\lambda} \frac{\phi^4}{4} + \delta C
\]
where $\delta Z$, $\delta M$ and $\delta \lambda$ are respectively the coefficients which express the wave function renormalization, the mass renormalization and the coupling constant renormalization. Furthermore $\delta C$ is constant in the field, see e.g. \cite{Collins} and it is inserted in order to render the contributions with external legs finite.
These parameters depend on the used regularization of the propagators, namely they depend on the parameter $\Lambda_2$ in the notation of the present paper, furthermore, they are divergent in the limit of vanishing $\Lambda_2$. They are chosen in such a way as to make the various Feynman graphs which appears in the perturbative expansion of the $S$-matrix well defined to all orders, see also \cite{BruDutFre09}.

We now observe that the Feynman graphs which need to be renormalized are those which have a superficial degree of divergence greater than or equal to $0$. 
In $\lambda \phi^4_4$, the superficial degree of divergence of a graph is $4-E$, where $E$ is the number of external lines and it does not depend on the number of internal vertices. Since the graphs with an odd number of external lines cancel because of symmetry reasons, we have that the graphs which needs to be regularized have $0$, $2$ or $4$ external legs.
The regularization of these graphs influences the renormalized constants $\delta Z$, $\delta M$ and $\delta \lambda$. 
In other words, contrary to the case $\lambda \phi^4_3$, we do not have an a priori closed expression for $\delta Z$, $\delta M$ and $\delta \lambda$ but we know them as asymptotic power series only because an infinite number of graphs need to be addressed. This is the first problem we have to tackle in applying the methods discussed in the present paper.

We furthermore observe that the wave function renormalization proportional to $\delta Z$ introduces  counterterms in  $L_I$ which are proportional to $\phi \Box \phi$ and so this interaction Lagrangian is not of the form given in \eqref{eq:interaction-Lagrangian} and analyzed in this paper.
This is the second problem which forbids a direct application of the methods discussed in section \ref{se:limit3d} for $\lambda \phi^4_4$.

We discuss now how to modify the description of the theory in order to put the problem in a form where Theorem \ref{th:convergence-Smatrix} and Theorem \ref{th:convergence-unitary1} can be applied.
To this end we assume that we have obtained the form of the counterterms to all orders in  the powers of the coupling constant in one of the available renormalization schemes.
Truncating the series at order $N$ we thus have the following functions
\[
\delta\lambda^{N}(\Lambda_2), \qquad \delta M^N(\Lambda_2), \qquad \delta Z^N(\Lambda_2), \qquad \delta C^N(\Lambda_2).
\]
In this paper we have seen that when the parameters $\Lambda_i$ are nonvanishing, the corresponding $S$-matrix is described by a well-defined unitary operator, we then want to construct sequences of unitary operators {which converge by a completeness argument in the weak-$*$ topology and such that the very same sequence applied to the power expansion of the corresponding $S$ matrix agrees with perturbation theory.}
%
%
The starting sequence is $\pi(S(\Lambda_{I,\Lambda_1/n, \Lambda_2/(\log(n))}))$ for various $n$,  we thus modify the definition of the Lagrangian making $N$ in the interaction Lagrangian
\[
L_{I,\Lambda_1,\Lambda_2,N} := \int e^{-\Lambda_1^2 \phi^2} (\delta C^N +\delta Z^{N} \phi\Box\phi - \delta M^{N} \phi^2 -\lambda^{N} \phi^4  ), 
\]
grow with $n$.
This solves the first problem.

To solve the second problem, we observe that the wave function renormalization can be treated with a redefinition of the free field theory together with a rescaling of the field configurations. 
In other words, {if $1+\delta Z^N$ is constant and positive,} the rescaled field is
\[
\phi_0 = \sqrt{1+\delta Z^N} \phi
\]
and with this rescaling, the free Lagrangian density of the theory is unchanged and equal to 
\[
\mathcal{L}_0(\phi_0) = -\frac{1}{2} \partial_\mu\phi_0\partial^\mu\phi_0 - m^2\frac{\phi_0^2}{2}.
\]
The obtained Lagrangian density, with respect to the rescaled field, with truncated renormalization constants at order $N$, is now:
\[
{
\mathcal{L}_{I} 
=  - \left(\frac{m^2+\delta M^{N}}{1+\delta Z^{N}}-m^2 \right) \frac{\phi_0^2}{2} -\left(\frac{\lambda+\delta \lambda^{N}}{(1+\delta Z^N)^2} \right)   \frac{\phi_0^4}{4}. }
\]
Here, as noticed above, $\delta Z^N$, $\delta \lambda^N$ and $\delta M^N$ are known for every $N$ and every $\Lambda_2$.

We introduce the following new functions of $N$ and the cutoff constant $\Lambda_2$
\begin{align*}
\tilde{M}(N,\Lambda_2) &:= \frac{m^2+\delta M^{N}}{1+\delta Z^{N}}-m^2\\
\tilde{\lambda}(N,\Lambda_2) &:= 
{ \frac{\lambda+\delta \lambda^{N}}{(1+\delta Z^N)^2}}\\
\tilde{C}(N,\Lambda_2) &:= {\delta C^{N}}. 
\end{align*}
%
%
With these at our disposal, {rescaling $\Lambda_1^2$ to $\Lambda_1^2/(1+\delta Z)$}, the interaction Lagrangian we have obtained is 
\begin{equation}\label{eq:LIPphi44}
L_{I,\Lambda_1,\Lambda_2,N}(g) :=  \int_{M}
e^{-\Lambda_1^2 \phi_0^2(x)}  \left( \tilde\lambda(N) \frac{ \phi_0^4(x)}{4} + \tilde{M}(N,\Lambda_2)(x) \frac{\phi_0^2(x)}{2}  +\tilde{C}(N,\Lambda_2)(x) \right) g(x) \dvol x
\end{equation}
and, for every $N$, it takes a form of one of the Lagrangians studied above in \eqref{eq:interaction-Lagrangian}.
This procedure and the identifications of relevant renormalization constants are called in the literature essential renormalization.

The free interaction Lagrangian $\mathcal{L}_0(\phi_0)$ is unchanged. We thus have that for fixed $N$ and fixed $\Lambda_i$,  $S(L_{I, \Lambda_1,\Lambda_2,N})$ can be obtained by applying Theorem \ref{th:convergence-Smatrix}. 
Furthermore by Theorem \ref{th:convergence-unitary1},
$\pi (S(L_{I, \Lambda_1,\Lambda_2,N}))$ is a unitary operator in the GNS representation of the Minkowski vacuum of the free theory $(\mathcal{H},\pi,\Psi)$.
The removal of the cutoffs can be obtained by constructing sequences of unitary operators for vanishing $\Lambda_i$ and growing $N$ which are convergent in the weak-$*$ topology over $\mathcal{B}(\mathcal{H})$. 
%

This can be done {keeping agreement with perturbation theory at every order in a way similar to that used  in Proposition \ref{eq:th-*weak-convergence} for the three-dimensional case.
In particular, to obtain agreement with perturbation theory, we observe that the regularised propagators are bounded by $1/\Lambda_2^2$
 and that in each Feynman diagram there are finitely many propagators. Furthermore, the counterterms $\delta Z^N, \delta M^N, \delta \lambda^N, C^N$ are used to regularize certain Feynman diagrams, hence, up to a finite reminder, in the limit of vanishing $\Lambda_2$ they can be bounded by the bounds satisfied by the Feynman diagrams they regularize. Hence they also grow at most polynomially in $1/\Lambda_2$. 
With this observation, we can use an argument similar to those in the proofs of Proposition \ref{th:3d-limits-noLambda} and of Corollary 
\ref{cor:perturbation-equal}.
%
}

As for the $\lambda\phi^4_3$ case, the limit point obtained {with the abstract argument
 could depend on the particular convergent subsequence of $I$. In fact, the existence of different subsequences of $I$ with different limit points cannot be excluded.
Despite this fact, also in the four-dimensional case the very same limits applied to the coefficients of the power expansion of $S$ should agree and coincide with the perturbative expansion.
}
%
%

\section{Towards the adiabatic limits: Mayer expansion in the regularized theory}\label{se:adiabatic-limits}

In this section we discuss the adiabatic limit of the expectation values of certain observables. 

{ 
Usually, the adiabatic limit is the limit where the interaction Lagrangian turns out to be nonvanishing everywhere in space, namely the limit where $g$ tends to $1$ taken in a suitable sense.
In the relativistic case the existence of this limit is forbidden by the Haag theorem. We can nevertheless consider 
the limits un which the interaction Lagrangian is constant either in space or in time. 
More precisely, 
 we can consider the limits where we take the smearing function to have the form $g(t,x)=\eta(t) h(x)$ and the spatial and temporal contributions are factorised.}

Furthermore, in this section we shall work in the theory where the cutoffs $\Lambda_i$ are strictly positive.

This is not done at the level of the $S$-matrix because we already know from Haag's Theorem that this limit cannot give a well-defined operator there. 
In particular, coefficients of the expansion of $S$ in powers of the coupling constant $\lambda$ diverge in that limit. 

At the algebraic level and in perturbation theory, the adiabatic limits of the Bogoliubov map can be treated using Einstein causality and at least for local observables this gives rise to a well-defined algebra of observables \cite{BrunettiFredenhagen00, HollandsWald2001}.

In view of the fact that causal factorisation holds in the time direction (see Proposition 4.3 in \cite{Doplicher:2019qlb}) and of the fact that the employed cutoff $\chi$ in $G_\Lambda$ is of compact support, the very same result applies in this paper.

\subsection{Adiabatic limit and Mayer expansion}

To analyze the long-distance behavior (adiabatic limit) at the level of $S$ matrices we may use Penrose Ruelle theorem \cite{Penrose, Ruelle} to bound the Mayer series independently of the {spatial} volume, see the lecture notes \cite{Procacci} and references therein for further details. Here we adapt those results to the case studied in this paper. 
In order to study the adiabatic limit, it is useful to study directly 
\[
\mathcal{P}(L_I) := \frac{1} {\text{vol}}\log(|S_\Lambda(L_I)(\phi)|) 
\]
where in this section $w = \Delta_{+,\Lambda}$ and where  $\log$ is taken after evaluating on a generic field configuration $\phi\in\mathcal{C}$.
Furthermore, $\text{vol} = \|g\|_1$.
In this section we shall consider 
\[
L_I= \lambda\int f(a)e^{\mathrm{i}\phi(x)a} \dvol a g(x) \dvol x.
\]
where $f$ is smooth and of compact support and its support is contained in the interval $[-A,A]$. $g$ is a compactly supported smooth function and we want to take the limit where $g$ tends to $1$ for $\left.\mathcal{P}(L_I)\right|_{\phi=0}$.

In order to obtain good estimates for $\mathcal{P}(L_I)$ and hence for 
$\log{|S_\Lambda|}$, 
we observe that for positive $\Lambda$
\[
|S_\Lambda(\phi)| \leq \tilde{S} 
\]
where now
\[
\tilde{S} = 
1+\sum_{n\geq 1} \frac{ \lambda^n }{n!} \int_{M^{n} } \dvol x^n \int_{\mathbb{R}^n}\dvol a^n 
\prod_{\tilde{j}=1}^{n} 
\left( g(x_{\tilde{j}}) |f(a_{\tilde{j}})|
\right) e^{-\sum_{i<j}a_ia_j w_s(x_i,x_j)}
\]
where $w_s=\Re{ \Delta_{+,\Lambda}}=\Re{ \Delta_{F,\Lambda}}$.
This expression is very close to 
a classical partition function. 
Actually, up to the integrals over $a$, 
$\tilde{S}$ is the grand-canonical partition function of a gas of classical particles with unfixed charge subjected to a two-body force descending from the classical potential $w_s$.
The main differences are in the fact that $w_s$ depends also on time and furthermore, there is no sharp choice of the charge $a$ and integrals over the charges need to be employed.
For these reasons the results already present in the literature cannot be applied straightforwardly and there is the need to adapt them.
This is the task we accomplish below.
We thus study 
\[
\tilde{\mathcal{P}}=\frac{1}{\|g\|_1\|f\|_1}\log{\tilde{S}} 
\]
which can be used to get upper bounds for $\mathcal{P}$ uniformly in $\phi$.
We expand it in powers of the coupling constant $\lambda$
\[
\tilde{\mathcal{P}} = \sum_{n\geq 1} \tilde{C}_n \lambda^{n}
\]
In the context of statistical mechanics, this power expansion is called {\bf Mayer expansion} and $\tilde{C}_n$ are called Mayer coefficients \cite{Procacci}. 
We recall that $\tilde{\mathcal{P}}(L_I)$ admits an expansion as a sum over graphs, to present it we recall a few basic facts about graphs.
We recall that a graph $G$ is formed by a couple $(V,E)$ where $V$ is the set of vertices usually depicted as points, and $E$ is the set of edges usually depicted as lines which connect two vertices. An edge is thus characterized by a couple of vertices.
The set of vertices of a graph is denoted by $V(G)$ and the set of edges by $E(G)$.
An edge with equal endpoints is called a loop (or a tadpole in the physical literature).
A graph has multiple edges if in its edge set there are multiple lines which have the same endpoints. 
A graph which has neither loops nor multiple edges is called simple.  A graph is connected if for every couple of vertices there exists a sequence of edges which joins the two vertices possibly passing through various other vertices. 

\begin{proposition}
The following holds
\[
\tilde{C}_n = \frac{1}{n!}\frac{1}{\|g\|_1\|f\|_1} \int\dvol x_1\dots \dvol x_n
\dvol a_1 \dots \dvol a_n
\prod_{{\tilde{j}}=1}^n \left(
g(x_{\tilde{j}})|f(a_{\tilde{j}})|
\right)
\Phi_T(x_1,a_1;\dots ; x_n,a_n)
\]
where for $n=1$ $\Phi_T(x_1,a_1)=1$ 
while for
$n>1$
\[
\Phi_T(x_1,a_1;\dots ; x_n,a_n) = \sum_{G\in\mathcal{G}_n} 
\prod_{\{i,j\}\in E(G)}  
\left(e^{-a_ia_jw_s(x_i,x_j)}
-1\right)
\]
where $\mathcal{G}_n$  is the set of connected simple graphs in  $\{1,\dots,n\}$. The edges $\{i,j\}\in E(G)$ are always ordered in such a way that $i<j$.
\end{proposition}

\begin{proof}
Notice that 
\begin{align*}
\tilde{S} &= 
1+\sum_{n\geq 1} \frac{\lambda^n }{n!} \int_{M^{n} } \dvol x^n \int_{\mathbb{R}^n}\dvol a^n 
\prod_{{\tilde{j}}=1}^{n} 
\left( g(x_{\tilde{j}}) |f(a_{\tilde{j}})|
\right) e^{-\sum_{i<j}a_ia_j w_s(x_i,x_j)}
\\
&= 
1+\sum_{n\geq 1} \frac{ \lambda^n }{n!}
\int_{M^{n} } \dvol x^n \int_{\mathbb{R}^n}\dvol a^n 
\prod_{{\tilde{j}}=1}^{n} 
\left(g(x_{\tilde{j}}) |f(a_{\tilde{j}})|
\right) 
\prod_{i<j}
\left(e^{-a_ia_j w_s(x_i,x_j)}-1+1\right)
\\
&= 
1+\sum_{n\geq 1} \frac{ \lambda^n }{n!} \int_{M^{n} } \dvol x^n \int_{\mathbb{R}^n}\dvol a^n 
\prod_{{\tilde{j}}=1}^{n} 
\left( g(x_{\tilde{j}}) |f(a_{\tilde{j}})|
\right) 
\sum_{G\in\tilde{\mathcal{G}}_n}
\prod_{\{i,j\}\in E(G)}
\left(e^{-a_ia_j w_s(x_i,x_j)}-1\right)
\end{align*}
where in the last equality we have expanded the product over couples of edges and 
where $\tilde{\mathcal{G}}_n$ is the set of all simple (unoriented) graphs (connected and not connected). $\tilde{\mathcal{G}}_n$ contains also the empty graph (the graph without edges). Its corresponding factor is $1$. Each graph in the sum over $\tilde{\mathcal{G}}_n$ is a collection of connected subgraphs over the elements of a suitable partition of $V(G)$. 
Hence expanding the sum over $\tilde{\mathcal{G}}_n$ as a sum over all possible partitions of its vertices and a sum over connected graphs we get
\begin{align*}
\tilde{S} 
&= 
\sum_{n\geq 0} \frac{ \lambda^n }{n!} \!\!\!\int_{M^{n} } \!\!\!\!\!\dvol x^n\!\! \int_{\mathbb{R}^n}\!\!\!\!\!\dvol a^n \!\!
\prod_{{\tilde{j}}=1}^{n} \!
\left( g(x_{\tilde{j}}) |f(a_{\tilde{j}})|
\right)\!\!
\sum_{k=1}^n \!\sum_{\{I_1,\dots, I_k\}\in\Pi_n^{k}}
\!
\prod_{l=1}^k
\sum_{G\in{\mathcal{G}}_{I_l}}
\!\prod_{\{i,j\}\in E(G)}
\!\!\!\!\left(e^{-a_ia_j w_s(x_i,x_j)}\!-\!1\right)
\\
&= 
1+\sum_{n\geq 1} \frac{ \lambda^n  }{n!} \int_{M^{n} } \dvol x^n \int_{\mathbb{R}^n}\dvol a^n 
\prod_{{\tilde{j}}=1}^{n} 
\left( g(x_{\tilde{j}}) |f(a_{\tilde{j}})|
\right)
\sum_{k=1}^n \sum_{\{I_1,\dots, I_k\}\in\Pi_n^{k}}
\prod_{l=1}^k
\Phi_T((x,a)_{I_l})
\end{align*}
where $\Pi_n^k$ is the set of all possible partitions of $n$ elements into $k$ nonempty subsets 
and $\mathcal{G}_{I}$ is the set of simple connected graphs joining the vertices contained in $I$. 
If $I=\{1\}$, namely it contains a single element then $\mathcal{G}_{\{1\}}$ is formed by  the empty graph only. In the previous formula the factor corresponding to the empty graph is $1$. 
Furthermore $(x,a)_{I_l} = (x_{I_l^1},a_{I_l^1};\dots ; x_{I_l^{|I|}},a_{I_l^{|I|}})$.
Distributing the integrals over 
$x^n$ and $a^n$ in the sum over possible partitions, 
we obtain that
\begin{align*}
\tilde{S} 
&= 
1+\sum_{n\geq 1} \frac{ \lambda^n }{n!} 
\sum_{k=1}^n \sum_{\{I_1,\dots, I_k\}\in\Pi_n^{k}}
\prod_{l=1}^k
\tilde{B}(|I_l|)
\end{align*}
where now 
\[
\tilde{B}(1) =  \int_{M } \dvol x \int_{\mathbb{R}}\dvol a\; 
 g(x) |f(a)|
\]
while for $k>1$
\[
\tilde{B}(k) = \sum_{G\in\mathcal{G}_k} \int_{M^{k} } \dvol x^k \int_{\mathbb{R}^k}\dvol a^k 
\prod_{{\tilde{j}}\in V(G)} 
\left( g(x_{\tilde{j}}) |f(a_{\tilde{j}})|
\right)
\prod_{\{i,j\}\in E(G)} \left(e^{-a_ia_j w_s(x_i,x_j)}-1\right).
\]
Notice that $\tilde{B}$  depends only on $k$, hence the sum over $k$, $n$ and the partitions $\Pi_n^k$ can be rearranged.
Actually, observing that 
the number of partitions in $\Pi_n^k$ of $n$ elements into $k$ subsets $\{I_j\}$ where the number of elements in each $I_j$ is fixed by $|I_j|$ is $n!/(|I_1|!\dots |I_k|!)$
and that in the rearrangement the same partitions $\{I_j\}$ and the corresponding $\prod_l \tilde{B}(|I_l|)$ can appear $k!$ times, we have that
\begin{align*}
\tilde{S} 
&= 
1+\sum_{n\geq 1} \frac{ \lambda^n }{n!} 
\sum_{k=1}^n \sum_{\{I_1,\dots, I_k\}\in\Pi_n^{k}}
\prod_{l=1}^k
\tilde{B}(|I_l|)
\\
&= 
1+\sum_{k\geq 1} \frac{1}{k!} 
\left(\sum_{m=1}^{\infty} \frac{ \lambda^m}{m!}  \tilde{B}(m)\right)^k
\\
&= \exp
\left(\sum_{m=1}^{\infty} \frac{ \lambda^m}{m!}  \tilde{B}(m)
\right).
\end{align*}
We have thus proved that  
\begin{align*}
\log\tilde{S} &=
\sum_{n=1}^{\infty} \frac{ \lambda^n}{n!} \tilde{B}(n)
\\
&=
\sum_{n=1}^{\infty} 
\frac{ \lambda^n}{n!}\int\dvol x_1\dots \dvol x_n
\dvol a_1 \dots \dvol a_n
\prod_{\tilde{j}} \left(
g(x_{\tilde{j}})|f(a_{\tilde{j}})|
\right)
\Phi_T(x_1,a_1;\dots ; x_n,a_n),
\end{align*}
from which the thesis follows.
\end{proof}
We may estimate the coefficients $\tilde{C}_n$ of the expansion of $\tilde{\mathcal{P}}$ in various ways.
If $f$ is of compact support and if $\supp f \subset [-A,A]$, we can bound 
$|e^{-a_ia_jw_s(x_i,x_j)}-1|$ 
by a constant $K$ because $w_s$ is continuous and $(x,a)$ are contained in a compact region, the Cartesian product of the support of $g$ with the support of $f$.  
Furthermore, if the support of $g$ and $f$ is sufficiently small, $K$ can be chosen to be $1$.
With these choices we have that 
\[
|\tilde{C}_n|  \leq \frac{(K\|g\|_1\|f\|_1)^{n-1}}{n!}
\sum_{G\in\mathcal{G}_n} 1 
\]
taking the sum over $\mathcal{G}_n$ we observe that we get an estimate which grows too strongly with $n$. The corresponding series cannot be proven to converge. Furthermore the bound depends on the support of $g$. 

In the next subsection we discuss how to remove this constraint.

\subsection{Kirkwood-Salsburg equations and their estimates}

In this section we discuss how to obtain an efficient bound for $\tilde{\mathcal{P}}$.
This discussion is an adaptation of the  theorems of Penrose and Ruelle which hold for a grand-canonical ensemble of particles.
We adapt to the setup of this paper the review presented in \cite{Procacci}. 
We start introducing the following set of correlation functions
\[
\rho_{n}(x_1,a_1;\dots ; x_n,a_n) := \frac{1}{\tilde{S}} \sum_{m\geq 0} \frac{\lambda^{m+n}}{m!} \int \dvol^m y \dvol^m b\; \prod_{{\tilde{j}}=1}^m g(y_{\tilde{j}}) |f(b_{\tilde{j}})|  e^{-U(x_1,a_1;\dots ; x_n,a_n;y_1,b_1;\dots ; y_m,b_m)}
\]
where 
\begin{align*}
U(x,a,y,b)
&=
U(x_1,a_1;\dots ; x_n,a_n;y_1,b_1;\dots ; y_m,b_m)
\\&= \sum_{1\leq i<j\leq n} a_ia_j w_s(x_i,x_j)
+\sum_{i=1}^n\sum_{j=1}^m a_ib_j w_s(x_i,y_j)
+\sum_{1\leq i<j\leq m} b_ib_j w_s(y_i,y_j).
\end{align*}
The expansion of $\rho_n$ as a power series is
\[
\rho_n(x_1,a_1;\dots ; x_n, a_n)
=
\sum_{m=0}^{\infty}\rho_{n,j}(x_1,a_1;\dots ; x_n, a_n) \lambda^{n+j}
\]
In this framework, we have the following proposition \cite{Procacci}
\begin{proposition}
The correlation functions
$\rho_{n}$ satisfy the Kirkwood-Salsburg equations
\begin{align*}
\rho_{n}(x_1,a_1;\dots ; x_n,a_n) 
=&  
\lambda e^{-W(x_1,a_1;x_2,a_2;\dots ;x_n,a_n)}
\\
&\cdot\sum_{s\geq 0}\frac{1}{s!}
\int \dvol y_1 g(y_1) \int \dvol b_1 |f(b_1)| 
\dots 
\int \dvol y_s g(y_s) \int \dvol b_s |f(b_s)|   
\\
&\cdot
\prod_{k=1}^s \left(e^{-a_1b_k w_s(x_1,y_k)}-1\right)
\rho_{n-1+s}(x_2,a_2;\dots ; x_n,a_n
;
y_1,b_1;\dots ; y_s,b_s).
\end{align*}
Where
\[
W(x_1,a_1;x_2,a_2;\dots ;x_n,a_n)
:=
\sum_{j=2}^n a_1a_j w_s(x_1,x_j).
\]
\end{proposition}
The perturbative expansion of these equations is such that 
\begin{equation}\label{eq:SK-pert}
\begin{aligned}
\rho_{n,\ell}(\tilde{x}_1,\dots,\tilde{x}_n)
&= e^{-W(\tilde{x}_1,\tilde{x}_2,\dots, \tilde{x}_n)}
\\
&\cdot
\sum_{s = 0}^\ell \frac{1}{s!} \int \dvol \mu_{\tilde{y}}^s \prod_{k=1}^s \left(e^{-a_1 b_k w_s(x_1,y_k)}-1\right)
\rho_{n-1+s,\ell-s}(\tilde{x}_2,\dots,\tilde{x}_n,\tilde{y}_1,\dots, \tilde{y}_s)
\end{aligned}
\end{equation}
where $\tilde{x}_j = (x_j,a_j)$
$\tilde{y}_j = (y_j,b_j)$
and 
$\dvol\mu_{\tilde{y}}^s = \prod_{j=1}^s g(y_j)  \dvol y_j |f(b_j)| \dvol b_j $.
Equations \eqref{eq:SK-pert} can be solved recursively, starting from 
\[
\rho_{n,0}(\tilde{x}_1,\dots, \tilde{x}_n) = e^{-U(\tilde{x}_1,\dots, \tilde{x}_n)}
\]
and 
\[
\rho_{0,\ell} = \delta_{0\ell}
\]
where here $\delta$ is the Kronecker delta.

There is a relation between the correlation function $\rho_n$ and the coefficients $\tilde{C}_\ell$. In particular, we have by direct inspection that 
\[
\int \rho_1(x,a) g(x)|f(a)| \dvol a \dvol x = \lambda\frac{d}{d\lambda} 
\log \tilde{S} 
= \|g\|_1\|f\|_1\sum_{\ell=1}^\infty \ell\tilde{C}_\ell   \lambda^{\ell}
\]
and this implies that
\begin{equation}
    \label{eq:tildeCl}
 \,\tilde{C}_{\ell+1}= \frac{1}{(\ell+1) \|g\|_1\|f\|_1}\int \rho_{1,\ell}(x,a) g(x)|f(a)| \dvol a \dvol x 
.   
\end{equation}
Having established the connection between the coefficients of the power expansion of $\log \tilde{S}$ and the correlation function $\rho$, we can now obtain the following estimates for the coefficients.

\begin{theorem}\label{th:estimate-adiabatic}
With $g\geq 0$ and $V$ of the form given above and $\supp{f(a)} \subset [-A,A]$ with $A$ sufficiently large
\[
|\tilde{C}_n| \leq  e^{2 B (n-2)}\frac{n^{n-2}}{n!} E^{n-1} 
\]
where {
\[
E = \sup_{a\in[-A,A]} \sup_{x\in\supp{g}} \int_{M,\mathbb{R}} \left|e^{-a b w_s(x,y) }-1 \right| g(y) |f(b)| \dvol y \dvol b
\]
and}
\[
{B} = 
\sup_{a\in[-A,A]} \sup_x  \frac{a^2} {2} w_s(x,x).
\]
\end{theorem}
\begin{proof}
This proof follows closely the original proof of Penrose and Ruelle \cite{Penrose, Ruelle}. 
We start discussing the bounds satisfied by the correlation functions.  
We look for suitable constants $K_{n,l}$ such that  
\begin{equation} \label{eq:bound}
|\rho_{n,l}(\tilde{x}_1,\dots, \tilde{x}_n)| \leq K_{n,l} 
\end{equation}
We observe that $|\rho_{1,0}|\leq K_{1,0}=1$; we then use the Kirkwood-Salsburg  recursive relations given in \eqref{eq:SK-pert} to obtain the other $K_{n,l}$. 
We can proceed by induction on $l$ to find the constant $K_{n,l-n}$. If we have obtained all the constants $K_{n,l-n}$ for $l\leq \tilde{l}-1$ and for $1\leq n\leq \tilde{l}$, we get 
\[
|\rho_{n,\tilde{l}-n} (\tilde{x}_1,\dots, \tilde{x}_n)| \leq e^{2 B} \sum_{s=0}^{n-\tilde{l}}\frac{1}{s!} E^s K_{n-1-s,\tilde{l}-n-s}
\]
here we have used the fact that $\rho_{n,\tilde{l}-n} (\tilde{x}_1,\dots, \tilde{x}_n)$ is invariant under permutations of the entries to show that there exists at least one permutation $j$ such that  
\begin{equation}\label{eq:step-permut}
e^{-W(\tilde{x}_{j_1},\dots, \tilde{x}_{j_{n}})}  \leq e^{2 B}
\end{equation}
This estimate can be obtained using an argument similar to the one used in the proof of Lemma \ref{le:estimate}. We start observing that $w_s$ is a positive product.
Denoting by $\Pi_n$ the set of all possible cyclic permutations of $\{1,\dots, n\}$ we have that 
$\sum_{j\in\Pi_n} W(\tilde{x}_{j_1},\dots, \tilde{x}_{j_n})=  \frac{w_s(\Psi, \Psi)}{2} - 2nB$ where $\Psi(y) = \sum_{i=1}^n  a_i \delta(y-x_i)$. Hence, 
$\sum_{l\in\Pi_n} W(\tilde{x}_{l_1}\dots, \tilde{x}_{l_n})\geq  - 2nB$, thus there exists at least one permutation $j$ such that the inequality \eqref{eq:step-permut} holds. 
We thus have that \eqref{eq:bound} holds for $K_{n,\tilde{l}-n}$ if 
\[
K_{n,\tilde{l}-n} = e^{2 B} \sum_{s=0}^{n-\tilde{l}}\frac{1}{s!} E^s K_{n-1-s,\tilde{l}-n-s}.
\]
A solution of these recursive relations for $K_{n,\tilde{l}-n}$ with $K_{1,0}=1$ is now
\[
K_{n,\tilde{l}}=e^{2B(n+\tilde{l}-1)}n(n+\tilde{l})^{\tilde{l}-1} \frac{E^{\tilde{l}}}{\tilde{l}!}.
\]
Having found these constants we have that
\[
|\rho_{1,l}(\tilde{x})| \leq e^{-2B} K_{1,l} =e^{2B(l-1)}(1+l)^{l-1} \frac{E^l}{l!}
\]
From this and from equation \eqref{eq:tildeCl} where the coefficients $\tilde{C}_{l+1}$ are related to $\rho_{1,l}$ we obtain the thesis.
\end{proof}

{\bf Remark}
We observe that if the state we are considering is the Minkowski vacuum then $w_s(x_1,x_2)$ used in the estimates of $E$ in Theorem \ref{th:estimate-adiabatic} is a function which decays exponentially for large spatial separations.
Furthermore, in the four-dimensional case, $w_s$ decays at least as $1/(t_1-t_2)^{3/2}$ in the temporal direction (thanks to the smearing provided by the cutoff $\chi$ used in the propagators, see e.g. \eqref{eq:Delta+}). 
{ 
While the limit where $g$ tends to $1$ in the whole spacetime is divergent because $w_s$ is a solution of the Klein-Gordon equation of mass $m$, the integral over the whole space keeping a cutoff in time for $d\geq 2$ or the integral over the whole time keeping a cutoff in space in the case of $d=4$ can still be taken. 
This can be done factorising the cutoff $g(t,x)=\eta(t)h(x)$ and taking the limit where $h\to1$ keeping $\eta$ fixed or taking the limit where $\eta\to1$ keeping $h$ fixed. 
%
%
With this observation we have that under one of these limits the constant $E$ remains finite.  
}

\bigskip

With the estimates given in Theorem \ref{th:estimate-adiabatic} we have that the sum which defines $\tilde{\mathcal{P}}$ converges if $\lambda$ is sufficiently small. 
The absolute convergence is actually established if $\lambda < \frac{1}{e^{2B+1}E}$. Furthermore, in view of the previous Remark, this estimate is uniform 
{ either in the spatial support of $g$ or in the temporal support of $g$ and thus improves the convergence of $\tilde{\mathcal{P}}$ in the adiabatic limits.}
%

\section*{Acknowledgments}
The author is grateful to Klaus Fredenhagen, Nicolò Drago and Simone Murro for many useful discussions about the subject of this paper. 
The author is  grateful for the support of the National Group of Mathematical Physics (GNFM-INdAM).
This research was supported in part by the MIUR Excellence Department Project 2023-2027 awarded to the Dipartimento di Matematica of the University of Genova, CUPD33C23001110001.

\printbibliography

{\bf Data availability:} 
Data sharing is not applicable to this article as no datasets were generated or analyzed during
the current study. 

\vspace{0.3cm}
{\bf Conflicts of interest statement:} 
The author has no  conflicts of interest.

\end{document}